\documentclass[10pt, journal, compsoc]{IEEEtran}

\usepackage{microtype}%if unwanted, comment out or use option "draft"

\usepackage{multicol}

 \usepackage{setspace}
 \usepackage[pdftex]{color,graphicx}
 \usepackage{subfigure}
 \usepackage[latin1]{inputenc}
\usepackage{algorithm}
\usepackage[noend]{algorithmic}
\usepackage{amssymb}
 \usepackage{times}
 \usepackage{pifont}

 \usepackage{theorem}

\newtheorem{thm}{Theorem}
\newtheorem{lem}{Lemma}
\newtheorem{assumption}{Assumption}
\newtheorem{definition}{Definition}

 \renewcommand{\algorithmiccomment}[1]{#1}
 \algsetup{
    linenosize=\scriptsize,
    linenodelimiter=)
 }

 \title{Efficient and Modular Consensus-Free Reconfiguration for Fault-Tolerant 
 Storage}
 \author{Eduardo Alchieri, Alysson Bessani, Fabiola Greve, Joni Fraga
 \\
 \small{University of Brasilia (Brazil), University of Lisbon (Portugal),}\\
 \small{Federal University of Bahia (Brazil), Federal University of Santa Catarina (Brazil).}
 \IEEEcompsocitemizethanks{
\IEEEcompsocthanksitem This work is supported by CNPq (Brazil) through project FREESTORE (CNPq 457272/2014-7).}
}

\begin{document}

\maketitle

\begin{abstract}
Quorum systems are useful tools for implementing consistent and available storage in the presence of failures.
These systems usually comprise a static set of servers that provide a fault-tolerant read/write register accessed by a set of clients.
We consider a dynamic variant of these systems and propose \textsc{FreeStore}, a set of fault-tolerant protocols that emulates a register in dynamic asynchronous systems in which processes are able to join/leave the servers set during the execution. 
These protocols use a new abstraction called \emph{view generators}, that captures the agreement requirements of reconfiguration and can be implemented in different system models with different properties.
% Particularly interesting, we present a reconfiguration protocol that is \emph{modular, efficient, consensus-free} and \emph{loosely coupled} with read/write protocols, making it easy to adapt other static fault-tolerant register implementation to dynamic environments.
Particularly interesting, we present a reconfiguration protocol that is \emph{modular, efficient, consensus-free} and \emph{loosely coupled} with read/write protocols, 
improving the overall system performance.
\end{abstract}

%\begin{IEEEkeywords}
%Dynamic Distributed Systems; Reconfiguration; Quorum Systems; Fault-Tolerant Storage
%\end{IEEEkeywords}

\section{Introduction}

Quorum systems \cite{Gif79} are a fundamental abstraction to ensure consistency and availability of data stored in replicated servers. 
Apart from their use as building blocks of synchronization protocols (e.g., consensus \cite{Cha96,Lam98}), quorum-based protocols for read/write (r/w) register implementation are appealing due to their promising scalability and fault tolerance: the r/w operations do not need to be executed in all servers, but only in a subset (quorum) of them.
The consistency of the stored data is ensured by the intersection between any two quorums of the system.

Quorum systems were initially studied in static environments, where servers are not allowed to join or leave the system during execution~\cite{Gif79,Att95}.
This approach is not adequate for long lived systems since, given a sufficient amount of time, there might be more faulty servers than the threshold tolerated, affecting the system correctness.
Beyond that, this approach does not allow a system administrator to deploy new machines (to deal with increasing workloads) or replace old ones at runtime.
Moreover, these protocols can not be used in many systems where, by their very nature, 
the set of processes that compose the system may change during its execution (e.g., MANETs and P2P overlays). 

Reconfiguration is the process of changing the set of nodes that comprise the system. 
Previous works proposed solutions for reconfigurable storage by implementing dynamic quorum systems \cite{Gil10,Mar04}.
Those proposals rely on consensus for reconfigurations, in a way that processes agree on the set of servers (view) supporting the storage.
Although adequate, since the problem of changing views resembles an agreement problem, this approach is not the most efficient or appropriate.
Besides the costs of running the protocol, consensus is known to not be solvable in asynchronous environments~\cite{Fis85}. 
Moreover, atomic shared memory emulation can be implemented in static asynchronous systems without requiring consensus~\cite{Att95}.

Until recently, the distributed systems community did not know whether it would be possible to implement reconfigurations without agreement.
Aguilera et al~\cite{Agu11} answered this question affirmatively with \emph{DynaStore}, a protocol that implements dynamic atomic storage without relying on consensus.
DynaStore reconfigurations may occur at any time, generating a graph of views from which it is possible to identify a sequence of views in which clients need to execute their r/w operations (see Figure~\ref{fig:comparison}, left).
Unfortunately, it presents two serious drawbacks: its reconfiguration and r/w protocols are strongly tied and its performance is significantly worse than consensus-based solutions in synchronous executions, which are the norm.
The first issue is particularly important since it means that DynaStore r/w protocols (as well as other consensus-based works like RAMBO~\cite{Gil10}) explicitly deal with reconfigurations and are quite different from static r/w protocols (e.g., the ABD protocol \cite{Att95}).

% Recently, the SmartMerge~\cite{Jen15} and SpSn~\cite{Gaf15} protocols improved DynaStore by separating the reconfiguration and r/w protocols. These approaches do not fully decouple r/w and reconfiguration protocols since before execute each 
% r/w operation it is necessary to access an abstraction implemented by a set of static registers to check for updates in the system.
% 

Recently, SmartMerge~\cite{Jen15} and SpSn~\cite{Gaf15} protocols improved DynaStore by separating the reconfiguration and r/w protocols. 
However, as in DynaStore, these approaches do not fully decouple de execution of r/w and reconfiguration protocols since before execute each 
r/w it is necessary to access an abstraction implemented by a set of static single-writer multi-reader SWMR registers to check for updates in the system. 
This design decision drops significantly the system performance, even in periods without reconfigurations.

In this paper we present \textsc{FreeStore}, a set of algorithms for implementing fault-tolerant atomic~\cite{Lam86} and wait-free~\cite{Her91} storage that allows the servers' set reconfiguration at runtime. 
\textsc{FreeStore} is composed by two types of protocols: (1) r/w protocols and (2) the reconfiguration protocol.
Read/write protocols can be adapted from static register implementations (e.g., the classical ABD~\cite{Att95}, as used in this paper).
The reconfiguration protocol -- our main contribution -- is used to change the set of replicas supporting the storage.
The key innovation of the \textsc{FreeStore} reconfiguration protocol is the use of \textit{view generators}, a new abstraction that captures the agreement requirements of reconfiguration protocols.
We provide two implementations of view generators, one based on consensus and other consensus-free, and compare how efficiently they can be used to solve reconfiguration.

%%%%%%%%%%%%%%%%%%%%%%%%%%%%%%%%%%%%%%%
\begin{figure}[!ht]
\begin{center}
%\subfigure[\footnotesize{Dynastore.}]{\label{fig:dyn}{
  \includegraphics[width=0.48\columnwidth]{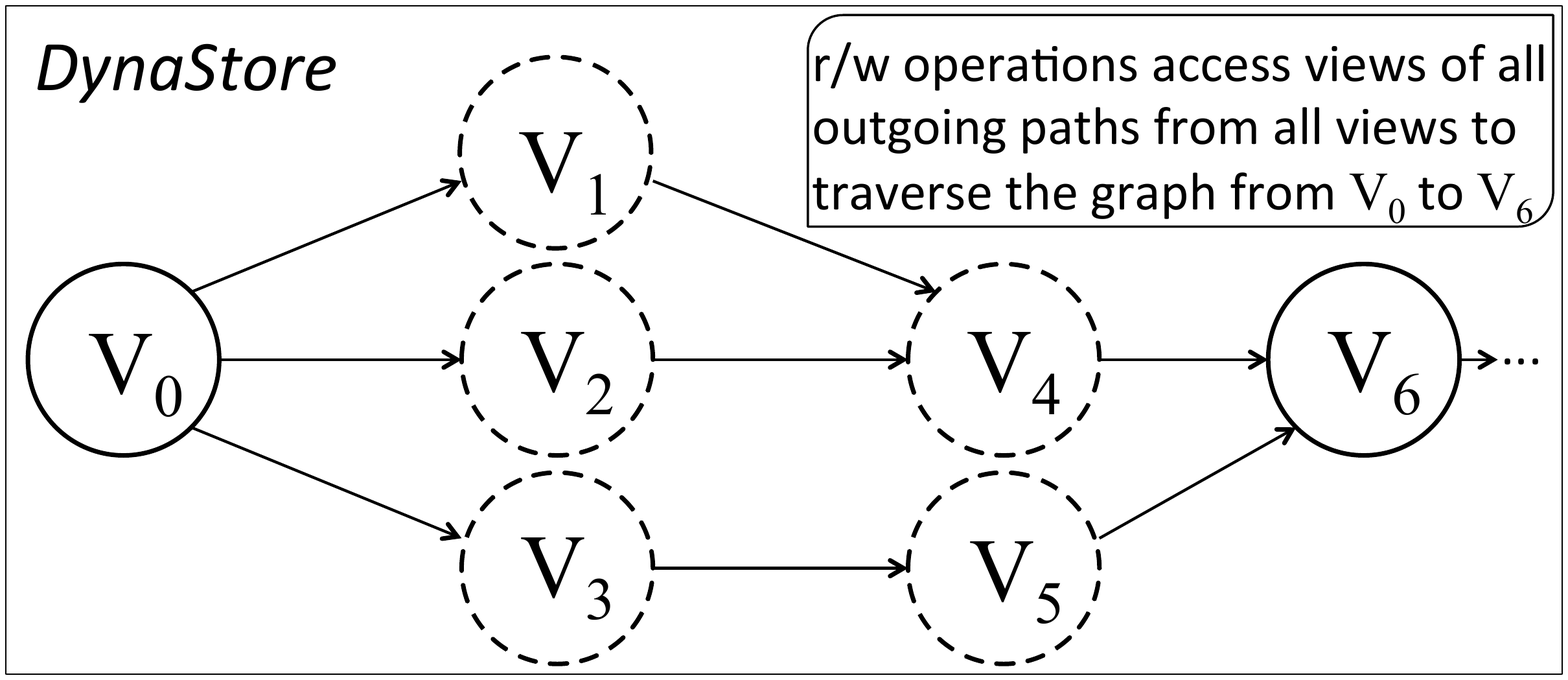}
%}}
%  \hspace{5mm}
%\subfigure[\footnotesize{\textsc{FreeStore}.}]{\label{fig:free}{
  \includegraphics[width=0.48\columnwidth]{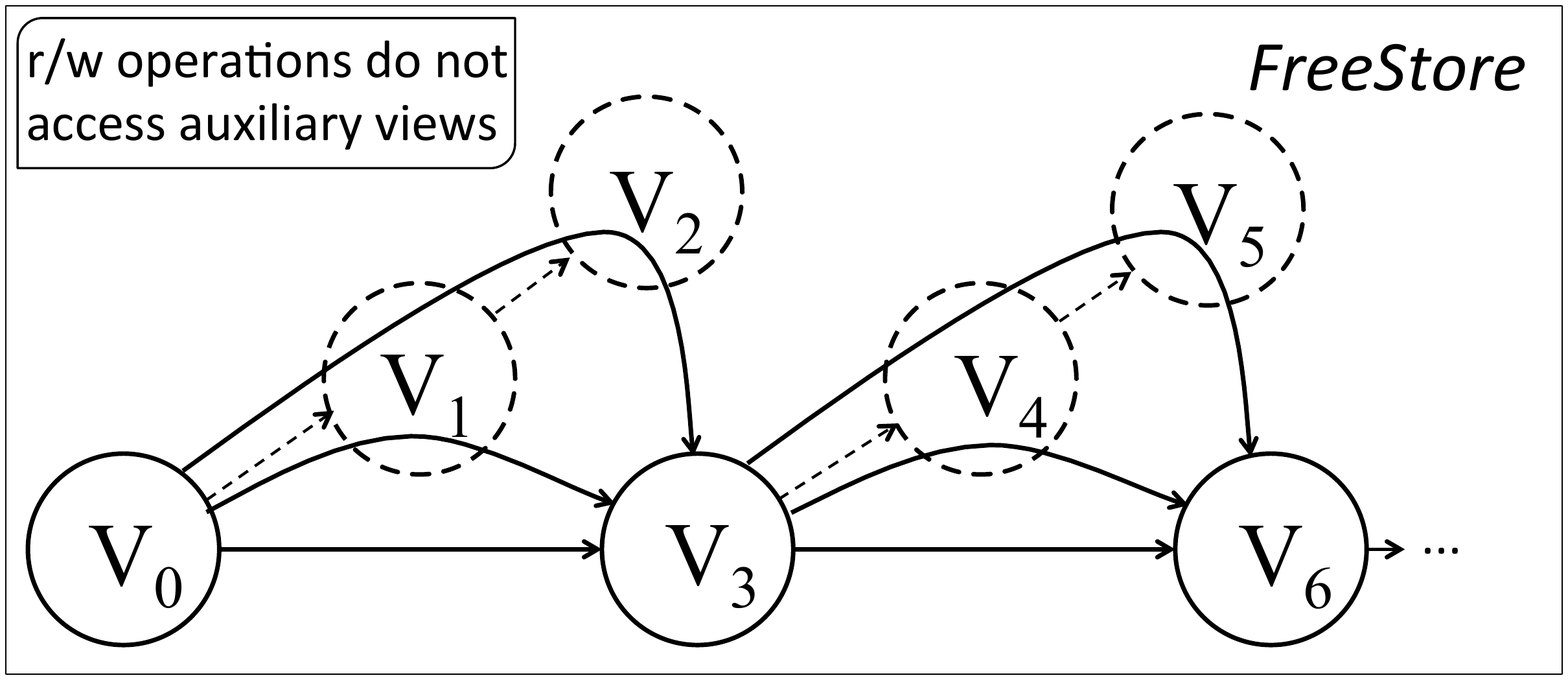}
%}}
%  \includegraphics[width=0.23\columnwidth]{figs/legend}
\end{center}
% \vspace{-6mm}
\caption{\small View convergence strategies of DynaStore~\cite{Agu11} (left) and \textsc{FreeStore} (right).
Dotted circles represent the auxiliary/non-established views that the system may experience during a reconfiguration. Solid circles represent the established/installed views process must converge to.}
\label{fig:comparison}
%  \vspace{-1mm}
\end{figure}
%%%%%%%%%%%%%%%%%%%%%%%%%%%%%%%%%%%%%%%

\textsc{FreeStore} improves the state of the art in at least three aspects: \emph{modularity}, \emph{efficiency} and \emph{simplicity}.
The modularity of the proposed protocols is twofold: it separates the reconfiguration from the r/w protocols, allowing other static storage protocols (e.g., \cite{Dut10,Fan03}) to be adapted 
to our dynamic model, and introduces the notion of view generators, capturing the agreement requirements of a reconfiguration protocol.
In terms of performance, both the r/w and reconfiguration protocols of \textsc{FreeStore} require less communication steps than their counterparts from previous works, either consensus-based~\cite{Gil10} or 
consensus-free~\cite{Agu11,Jen15,Gaf15}.
In particular, \textsc{FreeStore}' consensus-free reconfiguration requires less communication steps than other consensus-based reconfiguration protocols in the best case, matching the intuition that a 
consensus-free approach should be faster than another that rely on consensus.
% Finally, the consensus-free variant of \textsc{FreeStore} introduces a novel reconfiguration strategy that reduces the number of intermediary views to reach an installed view with all requested
Finally, the consensus-free variant of \textsc{FreeStore} introduces a novel reconfiguration strategy that reduces the number of intermediary views to reach an installed view with all requested 
updates (see Figure~\ref{fig:comparison}, right).
This strategy is arguably easier to understand than the one used in the previous works, shedding more light on the study of consensus-free reconfiguration protocols.

%\textsc{FreeStore} aims to bring consensus-free storage reconfiguration from a \emph{groundbreaking theoretical result} (which was perhaps too complex and inefficient to be used in practice) to a practical research topic that can originate novel fault-tolerant systems in the near future.
 
In summary, the main contributions of this paper are:

\begin{enumerate}

% \item  It introduces the notion of \emph{view generators}, an abstraction that captures the agreement requirements of a storage reconfiguration protocol and provides two implementations for such abstraction: the consensus-based \emph{perfect view generator} and the consensus-free \emph{live view generator};
\item  It introduces the notion of \emph{view generators}, an abstraction that captures the agreement requirements of a storage reconfiguration protocol and provides two implementations: the consensus-based \emph{perfect view generator} and the consensus-free \emph{live view generator};

\item It shows that \textit{safe} and \textit{live} dynamic fault-tolerant atomic storage can be implemented using the proposed view generators and discusses the tradeoffs between them;
% 
% \item It presents \textsc{FreeStore}, the first consensus-free dynamic atomic storage system in which the reconfiguration protocol is fully decoupled from the r/w protocols, increasing the system performance and 
%  making it easy to adapt other static fault-tolerant register implementation to dynamic environments;

\item It presents \textsc{FreeStore}, the first dynamic atomic storage system in which the reconfiguration protocol \textit{(a)}~can be configured to be consensus-free or consensus-based and 
 \textit{(b)}~is fully decoupled from the r/w protocols, increasing the system performance and 
 making it easy to adapt other static fault-tolerant register implementation to dynamic environments;

\item It presents a detailed comparison of \textsc{FreeStore} with previous protocols~\cite{Gil10,Agu11,Jen15,Gaf15}, showing that \textsc{FreeStore} is faster (requires less communication steps) 
than these previous systems.

\end{enumerate}

% \subsubsection{Paper organization}
% Section \ref{def-prelim} presents some preliminary definitions.
% Section \ref{geradores} introduces the view generators abstraction.
% Section \ref{secao_reconfiguracao} describes the \textsc{FreeStore} reconfiguration protocols.
% Section \ref{rw_prt} describes how classical r/w protocols can be adapted to work with our reconfiguration protocols.
% Section \ref{freestore_proofs} presents the correctness proofs for \textsc{FreeStore}.
% Finally, section \ref{discussao} discusses other interesting aspects of \textsc{FreeStore}.

% \vspace{-2mm}
\section{Preliminary Definitions}
\label{def-prelim} 

% \vspace{-1mm}
\subsection{System Model} 
\label{q_modelo}

We consider a fully connected distributed system composed by a universe of processes $U$, that can be divided in two non-overlapping subsets: an infinite set of servers $\Pi = \{1,2,...\}$; and an infinite set of clients $C = \{c_1,c_2,...\}$.
Clients access the storage system provided by a subset of the servers (a \emph{view}) by executing read and write operations.
Each process (client or server) of the system has a unique identifier.
Servers and clients are prone to \emph{crash failures}.
Crashed processes are said to be \emph{faulty}.
A process that is not faulty is said to be \emph{correct}.
%Process arrivals follow the infinite arrival model with bounded (and unknown) concurrency \cite{Agu04}.
Moreover, there are \emph{reliable channels} connecting all pairs of processes.

We assume an \emph{asynchronous distributed system} in which there are no bounds on message transmission delays and processing times.
However, each server of the system has access to a local clock used to trigger reconfigurations. 
These clocks are not synchronized and do not have any bounds on their drifts, being nothing more than counters that keep increasing.
Besides that, there is a real-time clock not accessed by the processes, used in definitions and proofs.

% \vspace{-2mm}
\subsection{Dynamic Storage Properties}
\label{dyn_def}

During a dynamic system execution, a sequence of views is installed to account for replicas' joins and leaves.
In the following we describe some preliminary definitions and the properties satisfied by \textsc{FreeStore}.

% \vspace{-3mm}
\paragraph*{Updates} We define $\mathit{update} = \{+,-\}$ $\times$ $\Pi$, where the tuple $\langle +,i \rangle$ (resp. $\langle -,i \rangle$) indicates that server $i$ asked to join (resp. leave) the system.
A \textit{reconfiguration} procedure takes into account updates to define a new system's configuration, which is represented by a \emph{view}.

%  \vspace{-3mm}
\paragraph*{Views} 
A view $v$ is composed by a set of \textit{updates} (represented by $\mathit{v.entries}$) and its associated membership (represented by $\mathit{v.members}$).
Consequently, $\mathit{v.members} = \{i \in \Pi \colon \langle +,i \rangle \in \mathit{v.entries} \wedge \langle -,i \rangle \not\in \mathit{v.entries}\}$.
To simplify the notation, we sometimes use $i \in v$ and $|v|$ meaning $i \in \mathit{v.members}$ and $|\mathit{v.members}|$, respectively.
Notice that a server $i$ can join and leave the system only once, but this condition can be relaxed in practice if we add an epoch number on each reconfiguration request.

We say a view $v$ is \emph{installed} in the system if some correct server $i \in v$ considers $v$ as its \emph{current view} and answers client r/w operations on this view. 
On the other hand, we say that the previous view, installed before $v$, was \emph{uninstalled} from the system.
At any time $t$, we define $V(t)$ to be the \emph{most up-to-date view} (see definition below) installed in the system.
We consider that $V(t)$ remains \textit{active} from the time it is installed in the system until \emph{all correct servers} of another most up-to-date view $V(t'), t' > t$, installs $V(t')$.

%  \vspace{-3mm}
\paragraph*{Comparing views}
We compare two views $v_1$ and $v_2$ by comparing their $entries$.
We use the notation $v_1 \subset v_2$ and $v_1 = v_2$ as an abbreviation for $v_1.entries \subset v_2.entries$ and $v_1.entries = v_2.entries$, respectively.
If $v_1 \subset v_2$, then $v_2$ is \textit{more up-to-date} than $v_1$.

%  \vspace{-3mm}
\paragraph*{Bootstrapping}  
We assume a non-empty initial view $V(0)$ which is known to all processes. 
At system startup, each server $i \in V(0)$ receives an initial view $v_0 = \{\langle +,j\rangle \colon j \in V(0)\}$.
% At system startup, each server $i \in V(0)$ receives a view $v_0 = \{\langle +,j\rangle \colon j \in V(0)\}$.

% \vspace{-4mm}
\paragraph*{Views vs. r/w operations}  
At any time $t$, r/w operations are executed only in $V(t)$.
When a server $i$ asks to \emph{join} the system, these operations are disabled on it until an \emph{enable operations} event occurs.
After that, $i$ remains able to process r/w operations until it asks to leave the system, which will happen when a \emph{disable operations} 
event occurs.%\footnote{A server that leaves the system can still receive messages, but it will ignore them since its r/w operations are disabled.}

% At any time $t$, a correct server $i$ will only execute r/w operations issued by clients if $i \in V(t)$. 
% Moreover, $V(t)$ is the only view in the system where clients execute read/write operations at time $t$.
% When a server $i$ asks to \emph{join} the system, these operations are disabled on it until an \emph{enable operations} event occurs.
% After that, $i$ remains able to answer r/w operations until it asks to leave the system, which will happen when a \emph{disable operations} event occurs.\footnote{A server that leaves the system can still receive messages, but it will ignore them since its operations are disabled.}

% \paragraph*{FreeStore Properties}  

%   \vspace{-2mm}
\begin{definition}[\textsc{FreeStore} properties]
\label{def_freestore}
\textsc{FreeStore} satisfies the following properties:

%\textsc{Storage properties}
\begin{itemize}
%  \vspace{-2mm}
\item \textbf{Storage Safety}: The r/w protocols satisfy the safety properties of an atomic r/w register~\cite{Lam86}.

%  \vspace{-2mm}
\item \textbf{Storage Liveness}: Every r/w operation executed by a correct client eventually completes.

%  \vspace{-2mm}
%OLDER: \item \textbf{Reconfiguration -- Join Safety}: If a correct server $i \not\in V(t)$ invokes the $\mathit{join}$ operation at time $t$, there will be some $t' > t$ such that $i \in V(t')$.
\item \textbf{Reconfiguration -- Join Safety}:
%(NOT SAFETY)If a correct server $i$ invokes the $\mathit{join}$ operation, then processes will install a view $v$ such that $i \in v$.
If a server $j$ installs a view $v$ such that $i \in v$, then server $i$ has invoked the $\mathit{join}$ operation or $i$ is member of the initial view.

%  \vspace{-2mm}
%OLDER: \item \textbf{Reconfiguration -- Leave Safety}: If a correct server $i \in V(t)$ invokes the $\mathit{leave}$ operation at time $t$, there will be some $t' > t$ such that $i \not\in V(t')$.
\item \textbf{Reconfiguration -- Leave Safety}:
%(NOT SAFETY) If a correct server $i$ invokes the $\mathit{leave}$ operation, then processes will install a view $v$ such that $i \not\in v$.
If a server $j$ installs a view $v$ such that $i \not\in v \wedge (\exists v' : i \in v' \wedge v' \subset v)$, then server $i$ has invoked the $\mathit{leave}$ 
operation.\footnote{We can relax this property and adapt our protocol (Section \ref{secao_reconfiguracao}) to allow that other process issues the leave operation on behalf of a crashed process.}

%  \vspace{-2mm}
\item \textbf{Reconfiguration -- Join Liveness}: Eventually, the \emph{enable operations} event occurs at every correct server that has invoked a $\mathit{join}$ operation.

%  \vspace{-2mm}
\item \textbf{Reconfiguration -- Leave Liveness}: Eventually, the \emph{disable operations} event occurs at every correct server that has invoked a $\mathit{leave}$ operation.

\end{itemize}
\end{definition}

% \vspace{-3mm}
\subsection{Additional Assumptions for Dynamic Storage}
\label{dyn_ass}
% 
% Dynamic fault-tolerant storage protocols require some additional assumptions~\cite{Gil10,Mar04,Agu11,Jen15,Gaf15}.
% In the following we state these assumptions.

Dynamic fault-tolerant storage protocols~\cite{Gil10,Mar04,Agu11,Jen15,Gaf15} require the following additional assumptions.

% \vspace{-1mm}
\begin{assumption}[Fault threshold]
\label{f_lim}
For each view $v$, we denote $v.f$ as the number of faults tolerated in $v$ and assume that $v.f \leq \lfloor \frac{|v.members|-1}{2}\rfloor$.
\end{assumption}
%\vspace{-2mm}

% \vspace{-4mm}
\begin{assumption}[Quorum size]
\label{q_size}
For each view $v$, we assume quorums of size $v.q = \lceil \frac{|v.members|+1}{2}\rceil$.
\end{assumption}
% \vspace{-2mm}

These two previous assumptions are a direct adaptation of the optimal resilience for fault-tolerant 
quorum systems~\cite{Att95} to account for multiple views and only need to hold for views present in the generated view sequences (see Section~\ref{geradores}).
% and need to hold while $v$ is active.

% \vspace{-2mm}
\begin{assumption}[Gentle leaves]
\label{g_leaves}
A correct server $i \in V(t)$ that asks to leave the system at time $t$ remains in the system until it knows that a more up-to-date view $V(t'), t'>t, i \not\in V(t')$ is installed in the system.  
\end{assumption}
% \vspace{-2mm}

The \emph{gentle leaves} assumption ensures that a \textit{correct} leaving server will participate in the reconfiguration protocol that installs the new view without itself. 
Thus, it cannot leave the system before a new view accounting for its removal is installed.
%If the server is faulty, nothing can be ensured.
% This is required for any dynamic system: departing replicas need to stay available to transfer their state to arriving replicas~\cite{Gil10,Mar04,Agu11,Jen15,Gaf15}. 
Other dynamic systems require similar assumption: departing replicas need to stay available to transfer their state to arriving replicas. 
Notice however that the fault threshold accounts for faults while the view is being reconfigured. 

% \vspace{-2mm}
\begin{assumption}[Finite reconfigurations]
\label{finite_rec}
The number of updates requested in an execution is finite.
\end{assumption}
% \vspace{-2mm}

This assumption is fundamental to guarantee that r/w operations always terminate.
As in other dynamic storage systems, it ensures that a client will restart phases of an r/w due to reconfigurations a finite number of times, and thus, eventually such operation will complete.
In practice, the updates proposed could be infinite, as long as each r/w is concurrent with a finite number of reconfigurations.

%\vspace{-1mm}
\section{View Generators}
\label{geradores}

\emph{View generators} are distributed oracles used by servers to generate sequences of new views for system reconfiguration.
This module aims to capture the agreement requirements of reconfiguration algorithms.
In order to be general enough to be used for implementing consensus-free algorithms, such requirements are reflected in the sequence of generated views, and not directly on the views.
This happens because, as described in previous works~\cite{Gil10,Agu11,Jen15,Gaf15,Agu10}, the key issue with reconfiguration protocols is to ensure that the sequence of (possibly conflicting) views generated during a reconfiguration procedure will converge to a single view with all requested view updates.

For each view $v$, each server $i \in v$ associates a view generator $\mathcal{G}^v_i$ with $v$ in order to generate a succeeding sequence of views (possibly with a single view).
Server $i$ interacts with a view generator through two primitives:
(1) $\mathcal{G}^v_i.\mathit{gen\_view(seq)}$, called by $i$ to propose a new view sequence $\mathit{seq}$ to update $v$; and
(2) $\mathcal{G}^v_i.\mathit{new\_view(seq')}$, a callback invoked by the view generator to inform $i$ that a new view sequence $\mathit{seq'}$ was generated for succeeding $v$.
An important remark about this interface is that there is no one-to-one relationship between $\mathit{gen\_view}$ and $\mathit{new\_view}$: a server $i$ may not call the first but receive 
several upcalls on the latter for updating the \emph{same view} (e.g., due to reconfigurations started by other servers).
However, if $i$ calls $\mathcal{G}^v_i.\mathit{gen\_view(seq)}$, it will eventually receive at least one upcall to $\mathcal{G}^v_i.\mathit{new\_view(seq')}$.% (see the Termination property bellow).

Similarly to other distributed oracles (e.g., failure detectors~\cite{Cha96}), view generators can implement these operations in different ways, according to the different environments they are designed to operate (e.g., synchronous or asynchronous systems).
However, in this paper we consider view generators satisfying the following properties.

\begin{definition}[\textsc{View Generators}] 
\label{def_view_gen}

A generator $\mathcal{G}^v_i$ (associated with $v$ in server $i$) satisfy the following properties: 
% \vspace{-5mm}

\begin{itemize}

\item  \textbf{Accuracy}: we consider two variants of this property:
% \vspace{-2mm}

\begin{itemize}

% \item \textbf{Strong Accuracy}: for all $i \in v$, there is a view sequence $\mathit{seq}$ such that any upcall $\mathcal{G}^v_i.\mathit{new\_view(seq_i)}$ generates a view sequence $\mathit{seq}_i = \mathit{seq}$.

%\item \textbf{Strong Accuracy}: if $i$ received the upcall $\mathcal{G}^v_i.\mathit{new\_view(seq_i)}$, then all servers $j \in v$ receive upcalls $\mathcal{G}^v_j.\mathit{new\_view(seq_j)}$ such that $\mathit{seq}_i = \mathit{seq_j}$.
% \vspace{-2mm}
\item \textbf{Strong Accuracy}: for any $i,j \in v$, if $i$ receives an upcall $\mathcal{G}^v_i.\mathit{new\_view(seq_i)}$ and $j$ receives an upcall $\mathcal{G}^v_j.\mathit{new\_view(seq_j)}$, then $\mathit{seq_i} = \mathit{seq_j}$.

% \item \textbf{Weak Accuracy}: for any two sequences $seq$ and $seq'$ produced by some $\mathcal{G}^v$, either $seq \subseteq seq'$ or $seq' \subset seq$.

%\item \textbf{Weak Accuracy}: if $i$ received the upcall $\mathcal{G}^v_i.\mathit{new\_view(seq_i)}$, then all servers $j \in v$ receive upcalls $\mathcal{G}^v_j.\mathit{new\_view(seq_j)}$ such that either $seq_i \subseteq seq_j$ or $seq_j \subseteq seq_i$.

\item \textbf{Weak Accuracy}: for any $i,j \in v$, if $i$ receives an upcall $\mathcal{G}^v_i.\mathit{new\_view(seq_i)}$ and $j$ receives an upcall $\mathcal{G}^v_j.\mathit{new\_view(seq_j)}$, then either $seq_i \subseteq seq_j$ or $seq_j \subset seq_i$.

\end{itemize}

% \vspace{-2mm}
\item  \textbf{Non-triviality}: for any upcall $\mathcal{G}^v_i.\mathit{new\_view(seq_i)}$, $\forall w \in \mathit{seq_i}$, $v \subset w$.

\item \textbf{Termination}: if a correct server $i \in v$ calls $\mathcal{G}^v_i.\mathit{gen\_view(seq)}$, then 
eventually it will receive an upcall $\mathcal{G}^v_i.\mathit{new\_view(seq_i)}$.

\end{itemize}
\end{definition}

Accuracy and Non-triviality are safety properties while Termination is related with the liveness of view generation. 
Furthermore, the Non-triviality property ensures that generated sequences contains only updated views.
% 
%The two types of Accuracy plus the satisfaction of Termination lead to the definition of four classes of view generators, as shown in Figure \ref{classes_geradores}.
%
%%%%%%%%%%%%%%%%%%%%%%%%%%%%%%%%%%%%%%%
%\begin{figure}[!ht]
% \centering
% \begin{footnotesize}
% \begin{tabular}{c|c|c}
%  \hline                & \multicolumn{2}{c}{\textsf{\textbf{Accuracy}}} \\ \cline{2-3}
%  \textsf{\textbf{Termination}} & \textsf{Strong}          & \textsf{Weak}   \\ \hline
%  \hline \textsf{Yes}        & \textit{Perfect}       & \textit{Live}   \\
%                                     & $\mathcal{P}$           & $\mathcal{L}$    \\
%  \hline \textsf{No}         & \textit{Strong} & \textit{Weak}         \\                               
%                                     & $\mathcal{S}$           & $\mathcal{W}$     \\  \hline
% \end{tabular}
% \end{footnotesize}
% \caption{View generators classes.}
% \label{classes_geradores}
%\end{figure}
%%%%%%%%%%%%%%%%%%%%%%%%%%%%%%%%%%%%%%%%
%
%Generators of classes $\mathcal{P}$ and $\mathcal{S}$ are considered safe because they always generate a single sequence of views.
%Generators of classes $\mathcal{P}$ and $\mathcal{L}$ are considered live because they always generate a sequence of view.
%The generators of class $\mathcal{W}$ are weak because they neither generate a single sequence of views nor ensure termination.
% 
Using the two variants of accuracy we can define two types of view generators: $\mathcal{P}$, the \emph{perfect view generator}, that satisfies \emph{Strong Accuracy} and $\mathcal{L}$, the \emph{live view generator}, that only satisfies \emph{Weak Accuracy}.
Our implementation of $\mathcal{P}$ requires consensus, while $\mathcal{L}$ can be implemented without such strong synchronization primitive.

\subsection{Perfect View Generators -- $\mathcal{P}$} 
\label{ger_perfect}

\textit{Perfect view generators} ensure that a single sequence of views will be eventually generated at all servers of $v$ for updating such view.
Our implementation for $\mathcal{P}$ (Algorithm~\ref{alg_perfect}) 
uses a deterministic Paxos-like consensus protocol~\cite{Lam98} that assumes a partially synchronous system model~\cite{Dwo88}.
The algorithm is straightforward:
when $\mathcal{G}^v_i.\mathit{gen\_view(seq)}$ is called by server $i$, it first verifies if $\forall w \in seq: v \subset w$ and, if it is the case, propose $seq$ as its value for 
the consensus reconfiguring $v$.
This verification aims to ensure that Non-triviality is respected, i.e., only updated views will be proposed. 
When a value $seq'$ is decided in the consensus associated with view $v$, an upcall to $\mathcal{G}^v_i.\mathit{new\_view(seq')}$ is executed.
The Termination and Strong Accuracy properties comes directly from the Agreement and Termination properties of the underlying consensus algorithm~\cite{Cha96,Lam98}.

%Under the Paxos framework, any server of $v$ can be a proposer (calling $\mathit{Paxos}^v.\mathit{propose}(seq)$), but all servers of $v$ are acceptors and learners 
%(they receive an upcall $\mathit{Paxos}^v.\mathit{learn}(seq')$, even if they did not propose anything).
%We also assume that once a value is learned for a Paxos instance associated with $v$, this value is locked and no other value will be learned in this algorithm.
%Notice that a server only proposes a sequence if the new views are strict extensions of $v$.
%
%%%%%%%%%%%%%%%%%%%%%%%%%%%%%%%%%%%%%%%%
\renewcommand{\algorithmiccomment}[1]{#1}
 \algsetup{
    linenosize=\scriptsize,
    linenodelimiter=)
 }

%  \vspace{-2mm}
\begin{algorithm}[ht!]
\caption{$\mathcal{P}$ associated with $v$ - server $i \in v$.}
\label{alg_perfect} 
% \begin{multicols}{2}
% \begin{spacing}{0.8}
%
\footnotesize{
%\scriptsize{
\begin{algorithmic}[1]
\renewcommand{\algorithmicrequire}{\textbf{procedure}}
%\vspace{1mm}
% \renewcommand{\algorithmicrequire}{\textbf{upon}}
\REQUIRE $\mathcal{G}^v_i.\mathit{gen\_view(seq)}$ 
%\IF[\hfill // only updated views are proposed]{$\mathit{\forall w \in seq: v \subset w}$}
\STATE \textbf{if} $\mathit{\forall w \in seq: v \subset w}$ \textbf{then} $\mathit{Paxos}^v.\mathit{propose}(seq)$ %\COMMENT{\hfill //starts a consensus in $v$}
%\ENDIF
%\vspace{1mm}
\renewcommand{\algorithmicrequire}{\textbf{upon}}
\REQUIRE $\mathit{Paxos}^v.\mathit{learn}(seq')$
%\COMMENT{\hfill // when decides by $seq'$ ...}
\STATE $\mathcal{G}^v_i.\mathit{new\_view(seq')}$ %\COMMENT{\hfill // ... inform this to process.}
\end{algorithmic}
}
\end{algorithm}
%%%%%%%%%%%%%%%%%%%%%%%%%%%%%%%%%%%%%%%%

\paragraph*{Correctness} The following theorem proves that Algorithm~\ref{alg_perfect} satisfies the properties of $\mathcal{P}$.

% \vspace{-4.5mm}
\begin{thm} \label{thm_perfect}
Algorithm \ref{alg_perfect} implements a perfect view generator in a partially synchronous system model.
\end{thm}

% \vspace{-1mm}
\begin{proof}
\textit{Termination} and \textit{Strong Accuracy} are ensured by classical consensus properties, namely,  Termination and Agreement, respectively.
\textit{Non-Triviality} is ensured by the verification of line 1 and the consensus Validity property (the decided value is one of the proposals).
\end{proof}

\subsection{Live View Generators -- $\mathcal{L}$}
\label{ger_live}

Algorithm~\ref{alg_live} presents an implementation for \textit{Live view generators} ($\mathcal{L}$).
Our algorithm does not require a consensus building block, being thus implementable in asynchronous systems.
On the other hand, it can generate different sequences in different servers for updating the same view.
We bound such divergence by exploiting Assumptions \ref{f_lim} and \ref{q_size}, which ensure that any quorum of the system will intersect in at least one correct server, making any generated sequence  for updating $v$ be contained in any other posterior sequence generated for $v$ (Weak Accuracy). 
Furthermore, the servers keep updating their proposals until a quorum of them converges to a sequence containing all proposed views (or combinations of them), possibly generating some intermediate sequences before this convergence.

To generate a new sequence of views, a server $i \in v$ uses an auxiliary function $\mathit{most\_updated}$ to get the most up-to-date view in a sequence of views (i.e., the view that is not contained in any other view of the sequence). 
Moreover, each server keeps two local variables: $\mathtt{SEQ}^v$ -- the last view sequence proposed by the server -- and $\mathtt{LCSEQ}^v$ -- the last sequence this server converged.

When server $i \in v$ starts its view generator, it first verifies (1) if it already made a proposal for updating $v$ and (2) if the sequence being proposed contains only updated views (line~1).
If these two conditions are met, it sends its proposal to the servers in $v$ (lines 2-3).

\vspace{-2.5mm}
%%%%%%%%%%%%%%%%%%%%%%%%%%%%%%%%%%%%%%%
\begin{algorithm}[ht!]

\caption{$\mathcal{L}$ associated with $v$ - server $i \in v$.}
\label{alg_live} 
% \begin{multicols}{2}
% \begin{spacing}{0.8}
\footnotesize{
%\scriptsize{
\begin{algorithmic}[1]

\renewcommand{\algorithmicrequire}{\textbf{functions:}}
\REQUIRE \COMMENT{Auxiliary functions}\\
$\mathit{most\_updated}(seq) \equiv w \colon (w \in seq) \wedge (\nexists w' \in seq \colon w \subset w')$  

\vspace{0.5mm}
\renewcommand{\algorithmicrequire}{\textbf{variables:}}
\REQUIRE \COMMENT{Sets used in the protocol}\\
$\mathtt{SEQ}^v \leftarrow \emptyset$ \hfill // proposed view sequence\\ 
$\mathtt{LCSEQ}^v \leftarrow \emptyset$ \hfill // last converged view sequence known\\

\renewcommand{\algorithmicrequire}{\textbf{procedure}}
\vspace{0.5mm}
\REQUIRE $\mathcal{G}^v_i.\mathit{gen\_view(seq)}$
\IF{$\mathtt{SEQ}^v = \emptyset$ $\wedge$  $\mathit{ \forall w \in seq: v \subset w}$}
\STATE $\mathit{\mathtt{SEQ}^v \leftarrow seq}$
\STATE $\forall j \in \mathit{v}, send \langle\mbox{SEQ-VIEW}, \mathtt{SEQ}^v \rangle$ to $j$
\ENDIF

\vspace{0.5mm}
\renewcommand{\algorithmicrequire}{\textbf{upon receipt of}}
\REQUIRE $\mathit{\langle\mbox{SEQ-VIEW},seq \rangle}$ from $j$
\IF{$\mathit{\exists w \in seq \colon w \not\in  \mathtt{SEQ}^v}$}  
%  \IF[\hfill //there are \textit{conflicting views} in the sequences]{$\mathit{\exists w, w' \colon w \in seq \wedge w' \in \mathtt{SEQ}^v \wedge w \not\subset w' \wedge w' \not\subset w}$}
\IF{$\mathit{\exists w, w' \colon w \in seq \wedge w' \in \mathtt{SEQ}^v \wedge w \not\subset w' \wedge w' \not\subset w}$}
 %[\hfill // $lcseq$ is more up-to-date than $\mathtt{LCSEQ}^v$]
%   \IF{$\mathit{most\_updated}(\mathtt{LCSEQ}) \subset \mathit{most\_updated}(lcseq)$}
%    \STATE $\mathtt{LCSEQ} \leftarrow \mathit{lcseq}$
%   \ENDIF
   \STATE $w \leftarrow \mathit{most\_updated}(\mathtt{SEQ}^v)$
   \STATE $w' \leftarrow \mathit{most\_updated}(seq)$
   \STATE $\mathit{\mathtt{SEQ}^v \leftarrow \mathtt{LCSEQ}^v \cup \{w.entries \cup w'.entries \}}$ 
%  \ELSE[\hfill // the sequences can be composed in a new sequence]
 \ELSE
  \STATE $\mathit{\mathtt{SEQ}^v \leftarrow \mathtt{SEQ}^v \cup seq}$
 \ENDIF
 \STATE $\forall k \in \mathit{v}, send \langle\mbox{SEQ-VIEW}, \mathtt{SEQ}^v\rangle$ to $k$
\ENDIF

\vspace{0.5mm}
\renewcommand{\algorithmicrequire}{\textbf{upon receipt of}}
\REQUIRE $\mathit{\langle\mbox{SEQ-VIEW},\mathtt{SEQ}^v\rangle}$ from $v.q$ servers in $v$
\STATE $\mathtt{LCSEQ}^v  \leftarrow \mathtt{SEQ}^v$
\STATE $\forall k \in \mathit{v}, send \langle\mbox{SEQ-CONV}, \mathtt{SEQ}^v\rangle$ to $k$

\vspace{0.5mm}
\renewcommand{\algorithmicrequire}{\textbf{upon the receipt of}}
\REQUIRE $\mathit{\langle\mbox{SEQ-CONV},seq' \rangle}$ from $v.q$ servers in $v$
\STATE $\mathcal{G}^v_i.\mathit{new\_view(seq')}$
\end{algorithmic}
}
 \vspace{-1mm}

\end{algorithm}
%%%%%%%%%%%%%%%%%%%%%%%%%%%%%%%%%%%%%%%
\vspace{-2mm}

Different servers of $v$ may propose sequences containing different views and therefore these views need to be organized in a sequence.
When a server $i \in v$ receives a proposal for a view sequence from $j \in v$, it verifies if this proposal contains some view it did not know yet (line 4).\footnote{Notice that it might happen that a server receives this message even if its view generator was not yet initialized, i.e., $\mathtt{SEQ}^v = \emptyset$.}
If this is the case, $i$ updates its proposal ($\mathtt{SEQ}^v$) according to two mutually exclusive cases:

\begin{itemize}

\item
\textsc{Case 1} [There are \textit{conflicting views} in the sequence proposed by $i$ and the sequence received from $j$ (lines 5-8)]:
In this case $i$ creates a new sequence containing the last converged sequence ($\mathtt{LCSEQ}^v$) and a new view with the union of the two most up-to-date conflicting views. 
This maintains the containment relationship between any two generated view sequences.

\item
\textsc{Case 2} [The sequence proposed by $i$ and the received sequence \emph{can be composed in a new sequence} (lines 9-10)]:
In this case the new sequence is the union of the two known sequences.

\end{itemize}

In both cases, a new proposal containing the new sequence is disseminated (line 11).
When $i$ receives the same proposal from a quorum of servers in $v$, it \emph{converges} 
to $\mathtt{SEQ}^v$ and stores it in $\mathtt{LCSEQ}^v$, informing other servers of $v$ about it (lines 12-13). 
When $i$ knows that a quorum of servers of $v$ converged to some sequence $seq'$, $i$ \emph{generates} $seq'$ (line~14).

\paragraph*{Correctness}
The algorithm ensures that if a quorum of servers converged to a sequence $seq'$ (lines 12-13), then
\textit{(1)} such sequence will be generated (line 14) and
\textit{(2)} any posterior sequence generated will contain $seq'$ (lines 8 and 10), ensuring Weak Accuracy.
This holds due to the quorum intersection: at least one correct server needs to participate in the generation of 
both sequences and this server applies the rules in Cases 1 and 2 to ensure that sequences satisfy the containment relationship.
The servers in $v$ generate at most $|v|-v.q+1$ sequences of views, such that  $seq_1 \subset ... \subset seq_{|v|-v.q+1}$.
The sequence $seq_1$ is obtained by merging the proposals of a quorum of servers; the other 
sequences may be built using $seq_1$ plus different proposals from each of the $|v|-v.q$ servers that did not participate in the generation of $seq_1$.
The Termination property is ensured by the fact that \textit{(1)} each server makes at most one initial proposal (lines 1-3); 
\textit{(2)} servers keep updating their proposals until a quorum agree on some proposal; and \textit{(3)} there is always a quorum of correct servers in $v$.

The following lemmata and theorem prove that Algorithm~\ref{alg_live} satisfies the properties of $\mathcal{L}$.

%%%%%%%%%%%%%%%% PROVA %%%%%%%%%%%%%%%%%%%
\begin{lem} \label{lemma_vivo_weakacc}\vspace{-0mm}
(Weak Accuracy) 
Consider the view generators $\mathcal{G}^{v}$ implemented by Algorithm \ref{alg_live} associated with view $v$. 
For any $i,j \in v$, if $i$ receives an upcall $\mathcal{G}^v_i.\mathit{new\_view(seq_i)}$ and 
$j$ receives an upcall $\mathcal{G}^v_j.\mathit{new\_view(seq_j)}$, then either $seq_i \subseteq seq_j$ or $seq_j \subset seq_i$.
\end{lem}
 
 \vspace{-1mm}
 \begin{proof}
%Consider the view generators $\mathcal{G}^v$ implemented by Algorithm \ref{alg_live} associated with view $v$.
Consider that $\mathcal{G}^v$ generates sequences $seq_i$ and $seq_j$ at servers $i, j \in v$ by calling $\mathcal{G}^v_i.\mathit{new\_view(seq_i)}$ and $\mathcal{G}^v_j.\mathit{new\_view(seq_j)}$, respectively. 
If $\mathcal{G}^v_i$ generated $seq_i$, there is a quorum of servers in $v$ that sent $\mbox{SEQ-VIEW}$ messages with $seq_i$ (lines 12-13). 
The same can be said about the generation of $seq_j$ by $\mathcal{G}^v_j$: a quorum of servers in $v$ 
proposed $seq_j$ in $\mbox{SEQ-VIEW}$ messages.
Due to the intersection between quorums, there will be at least one correct server $k$ that sent $\mbox{SEQ-VIEW}$ messages both for $seq_i$ and $seq_j$.
Consider that (1) server $k$ first sent $seq_i$, (2) $i$ receives a quorum of these proposals and converges to $seq_i$ and (3) $j$ made an initial proposal $seq^{p}_j$.
Then, we have three cases:

\begin{enumerate}

\item  Case $seq_i = seq^{p}_j$, we have two cases:
\textit{(i)} $j$ receives the same messages that made $i$ converge to $seq_i$ and $\mathcal{G}^v_j$ generates $seq_j = seq_i = seq^{p}_j$ (lines 12-14); or
\textit{(ii)} $j$ receives a proposal containing other views that are not in $seq^{p}_j$ and updates its proposal before the generation of $seq_j$ (lines 8 and 10).
In \textit{(ii)}, since $seq_i$ was generated, at least a quorum of servers in $v$ store $seq_i$ in their $\mathtt{LCSEQ}^v$ (line 12), which means $seq_j$ must contain $seq_i$.

\item Case $seq_i \neq seq^{p}_j$ and \emph{conflicting}, the reception of $\mbox{SEQ-VIEW}$ from $k$ makes $j$ proposes a new sequence containing its $\mathtt{LCSEQ}^v$ ($lcseq_j$) 
plus a new view $w$, representing the union of all update requests (line 8): $lcseq_j \rightarrow w$.
Once $seq_i$ was generated, at least a quorum of servers in $v$ store $seq_i$ in their $\mathtt{LCSEQ}^v$ (line 12).
Since $\mathtt{LCSEQ}^v$ is updated only if $\mbox{SEQ-VIEW}$ messages with the same sequence are received from a quorum of servers (line 12), it will not contain conflicting sequences in different servers.
Consequently, the generated sequence $seq_j$ must contain $seq_i$.

\item  Case $seq_i \neq seq^{p}_j$ and \emph{not conflicting}, $j$ updates its proposal to a sequence $seq_j$ representing the union of $seq_i$ and $seq^{p}_j$ (line 10).
After that, $j$ sends a $\mbox{SEQ-VIEW}$ message with $seq_j$ to other servers that eventually update their proposals to this sequence.
Consequently, $\mathcal{G}^v_j$ generated a $seq_j$ such that $seq_i \subset seq_j$.
\end{enumerate}

Therefore, in the three possible cases $seq_i \subseteq seq_j$. 
Using similar arguments it is possible to prove the symmetrical execution: if server $k$ sent a $\mbox{SEQ-VIEW}$ message 
first for $seq_j$ and then for $seq_i$, we have that $seq_j \subseteq seq_i$.
Consequently, we have that either $seq_i \subseteq seq_j$ or $seq_j \subset seq_i$.
% $\strut\hfill \Box$
\end{proof}

\begin{lem} \label{lemma_vivo_termination}\vspace{-0mm}
(Termination)
Consider the view generators $\mathcal{G}^{v}$ implemented by Algorithm \ref{alg_live} associated with view $v$. 
If a correct server $i \in v$ calls $\mathcal{G}^v_i.\mathit{gen\_view(seq)}$, then eventually it will receive an upcall $\mathcal{G}^v_i.\mathit{new\_view(seq_i)}$.
\end{lem}

\vspace{-1mm}
\begin{proof}
When $\mathcal{G}^v_i.\mathit{gen\_view(seq)}$ is called, server $i$ sends a message with its proposal $seq$ (line 3). 
Since the channels are reliable, messages sent by correct processes are received by correct processes (lines 4-11). 
For $\mathcal{G}^v_i$ to generate a new sequence of views $seq_i$ (line~14), at least $v.q$ servers should converge to $seq_i$ (lines 12-13). 
During this convergence (lines 4-11), $i$ changes its proposed sequence ($\mathtt{SEQ}^v$) with proposals received from other servers (line 11).
Since there are at least $v.q$ correct servers in $v$ and each server makes a single initial proposal (line 1), at least $v.q$ servers are going to converge to the same sequence. 
Thus, $i$ will eventually generate $seq_i$ by executing $\mathcal{G}^v_i.\mathit{new\_view(seq_i)}$ (line 14).
% $\strut\hfill \Box$
\end{proof}

\begin{lem} \label{lemma_num_seq}\vspace{-0mm}
(Bounded Sequences for $\mathcal{L}$)
Consider the view generators $\mathcal{G}^{v}$ implemented by Algorithm \ref{alg_live} associated with view $v$. 
The number of different view sequences generated by $\mathcal{G}^{v}$ is bounded.
\end{lem}

\vspace{-1mm}
\begin{proof}
$\mathcal{G}^{v}$ must collect a quorum of proposals for a sequence $seq$ to generate it (lines 12-14). 
Each server $i \in v$ makes a single initial proposal (line 1) and updates it only by processing proposals from other server of $v$ (lines 8 and 10). 
Consequently, $\mathcal{G}^{v}$ must process the proposals from at least a quorum of servers to generate $seq$, being  $|v| - v.q$ the number of servers whose proposals have not been processed to generate $seq$.
In the worst case, each of these servers make a different proposal and such proposal is combined with the last generated sequence to generate a new sequence.
Consequently, $\mathcal{G}^{v}$ generates at most $|v|-v.q+1$ sequences.
% $\strut\hfill \Box$
\end{proof}

\begin{thm} \label{thm_vivo_prop}\vspace{-0mm}
Algorithm \ref{alg_live} implements a live view generator.
\end{thm}

\vspace{-1mm}
\begin{proof}
\textit{Weak Accuracy} and \textit{Termination} follows directly from Lemmata \ref{lemma_vivo_weakacc} and \ref{lemma_vivo_termination}, respectively. 
The \textit{Non-Triviality} property is ensured by the verification of line 1 and the way further proposals are defined (lines 8 and 10).  
% $\strut\hfill \Box$
\end{proof}

% \vspace{-2mm}
\section{\textsc{FreeStore} Reconfiguration}
\label{secao_reconfiguracao}

A server running the \textsc{FreeStore} reconfiguration algorithm uses a view generator associated with its current view~$cv$ to process one or more reconfiguration requests (joins and leaves) that will lead the system from $cv$ to a new view $w$.
%All pending reconfiguration requests are processed periodically, in batches.
%In consequence, the reconfiguration algorithms can process more than one update in a single run and r/w operations can execute in parallel with reconfigurations. 
% 
Algorithm \ref{alg_recon} describes how a server $i$ executes reconfigurations. 
In the following sections we first describe how view generators are started and then we proceed to discuss the behavior of this algorithm when started with 
either $\mathcal{L}$ (Section \ref{rec_ger_vivo}) or $\mathcal{P}$ (Section \ref{rec_ger_perfeito}).
% In the following sections we first describe how view generators are started and then we proceed to discuss the behavior of this algorithm when started 
% with $\mathcal{L}$ (reconfigurations with $\mathcal{P}$ are very simple since processes agree in the sequence of views generated).

% \vspace{-2mm}
\subsection{View Generator Initialization}
\label{vg_init}

Algorithm \ref{alg_recon} (lines 1-7) describes how a server $i$ processes reconfiguration requests and 
starts a view generator associated with its current view $cv$. 
A server $j$ that wants to join the system needs first to find the current view $cv$ and then to execute the $\mathit{join}$ operation (lines 1-2), sending a 
tuple $\langle +,j \rangle$ to the members of $cv$.
Servers leaving the system do something similar, through the $\mathit{leave}$ operation (lines 3-4).
When $i$ receives a reconfiguration request from $j$, it first verifies if the requesting server is 
using the same view as itself; if this is not the case, $i$ replies its current view to $j$ (omitted from the algorithm for brevity).
If they are using the same view and $i$ did not execute the requested reconfiguration before, it stores the reconfiguration tuple in its 
set of pending updates $\mathtt{RECV}$ and sends an acknowledgment to $j$ (lines 5-6).

For the sake of performance, a local \textit{timer} has been defined in order to periodically process the updates requested in a view, that is, the next system reconfiguration.
A server $i \in cv$ starts a reconfiguration for $cv$ when its timer expires and $i$ has some pending reconfiguration requests (otherwise, the timer is renewed).
The view generator is started with a sequence containing a single view representing the current view plus the pending updates (line~7).

%The following lemma states this property.
%
%%%%%%%%%%%%%%%%% PROVA %%%%%%%%%%%%%%%%%%%
%\begin{lem} \label{lemma_iniciando}
%A server with pending updates will eventually start the view generator associated with its current view.
%\end{lem}
%
%\begin{proof}
%Assume that $v$ is a view installed in the system and that there are pending \textit{updates}, then $\mathtt{RECV} \neq \emptyset$ for some process $i \in v$.
%Additionally, the \textit{timeout} associated to $v$ is going to be expired. 
%Then, $i$ is going to activate the view generator associated with $v$ (line 7). 
%\end{proof}
%%%%%%%%%%%%%%%%% PROVA %%%%%%%%%%%%%%%%%%%

% \vspace{-2mm}
\subsection{Reconfiguration using $\mathcal{L}$}
\label{rec_ger_vivo}

This section presents the reconfiguration processing when Algorithm \ref{alg_recon} is instantiated with \textit{live view generators} that may generate different views at different servers.
In this case, the algorithm ensures that the system \emph{installs} an unique sequence of views. 
%The key idea of the algorithm is to install only some of the views in the sequences generated by $\mathcal{L}$ to define an unique \textit{sequence of views} until a single final view containing the requested updates is installed.

% \vspace{-2mm}
% \subsubsection{Overview}
\paragraph*{Overview}
% 
% After the generation of a view sequence by $\mathcal{L}$, views of this sequence are started, one after another, until the most up-to-date view is installed.
% Recall that $\mathcal{L}$ satisfies only Weak Accuracy, consequently, different sequences may be generated for updating the same view, but they will always be compatible, i.e., contained one in another.
% 
Given a sequence $\mathit{seq}: v_1  \rightarrow \ldots  \rightarrow v_k \rightarrow w$ generated for updating a view $v$, only the last view $w$ will be installed in the system.
The other $k$ \emph{auxiliary views} are used only as intermediate steps for installing $w$.
%If a server knows this whole sequence, these views will not be installed since a server installs only the most up-to-date view in a sequence.
The use of auxiliary views is fundamental to ensure that no write operation executed in any of the views of $\mathit{seq}$ (in case they are installed in some server) ``is lost'' by the reconfiguration processing.
This is done through the execution of a \emph{``chain of reads''} in which the servers of $v$ transfer their state to the servers of $v_1$, which transfer their state to the servers of $v_2$ and so on until the servers of $w$ have the most up-to-date state.
To avoid consistency problems, r/w operations are disabled during each of these state transfers.
%EDUARDO: tirei este footnote pq precisa mesmo atualizar o estado de todos, pois pode ser que um servidor que estÃ¡ no sistema esteja com o estado desatualizado!
% \footnote{In practice it is necessary to transfer the state only to the servers joining the system.}

It is important to remark that since we do not use consensus (nor solve an equivalent problem),  a subsequence $\mathit{seq'}: v_1  \rightarrow \ldots  \rightarrow v_j, j \leq k$, of $\mathit{seq}$ may lead to the installation of $v_j$ in some servers that did not know $seq$ and that these servers may execute r/w operations in this view.
However, the algorithm ensures these servers eventually will reconfigure from $v_j$ to the most up-to-date view $w$.
%This is only possible due to the Weak Accuracy of $\mathcal{L}$.

%%%%%%%%%%%%%%%%%%%%%%%%%%%%%%%%%%%%%%%
\begin{algorithm*}[!ht]
 
\caption{\textsc{FreeStore} reconfiguration - server \textit{i}.}
\label{alg_recon}
% \begin{multicols}{2}
% \begin{spacing}{0.8}
\footnotesize{
\begin{algorithmic}[1]
%\begin{multicols}{2}
  
 \renewcommand{\algorithmicrequire}{\textbf{functions:}}
\REQUIRE \COMMENT{Auxiliary functions}\\
$\mathit{least\_updated}(seq) \equiv w \colon (w \in seq) \wedge (\nexists w' \in seq \colon w' \subset w)$

\renewcommand{\algorithmicrequire}{\textbf{variables:}}
\REQUIRE \COMMENT{Sets used in the protocol}\\
 $cv \leftarrow v_0$ \hfill // the system current view known by $i$ \\
 $\mathtt{RECV} \leftarrow \emptyset$ \hfill // set of received updates
 %$\mathcal{G}^{cv}_i \leftarrow$ the view generator associated with $cv$

\vspace{0.5mm}
\renewcommand{\algorithmicrequire}{\textbf{procedure}}
\REQUIRE $\mathit{join()}$
\STATE $\forall j \in cv, \mathit{send \langle\mbox{RECONFIG}, \langle +,i\rangle, cv\rangle}$ to $j$
\STATE \textbf{wait} for $\mathit{\langle\mbox{REC-CONFIRM}\rangle}$ replies from $\mathit{cv.q}$ servers in $cv$

\vspace{0.5mm}
\REQUIRE $\mathit{leave()}$
\STATE $\forall j \in \mathit{cv}, \mathit{send \langle\mbox{RECONFIG}, \langle -,i\rangle, \mathit{cv}\rangle}$ to $j$
\STATE \textbf{wait} for $\mathit{\langle\mbox{REC-CONFIRM}\rangle}$ replies from $\mathit{\mathit{cv}.q}$ servers in $\mathit{cv}$

\vspace{0.5mm}
\renewcommand{\algorithmicrequire}{\textbf{upon receipt of}}
\REQUIRE $\langle\mbox{RECONFIG}, \langle *,j\rangle, \mathit{cv}\rangle$ from $j$ and $\langle *,j\rangle \not\in cv$
\STATE $\mathtt{RECV} \leftarrow \mathtt{RECV} \cup \{ \langle *,j\rangle \}$
\STATE $\mathit{send \langle\mbox{REC-CONFIRM}\rangle}$ to $j$

\vspace{0.5mm}
\renewcommand{\algorithmicrequire}{\textbf{upon}}
\REQUIRE ($\mathit{timeout}$ for $\mathit{cv}$) $\wedge (\mathtt{RECV} \neq \emptyset$) 
\STATE $\mathcal{G}^{cv}_i.\mathit{gen\_view}(\{cv \cup \mathtt{RECV}\})$

\vspace{0.5mm}
\REQUIRE $\mathcal{G}^{ov}_i.\mathit{new\_view(seq)}$ \COMMENT{\hfill //$\mathcal{G}$ generates a new sequence of views to update $ov$ (usually $ov = cv$)}
\STATE $w \leftarrow \mathit{least\_updated}(\mathit{seq})$ \COMMENT{\hfill //the next view in the sequence $seq$}
% \STATE \emph{R-multicast}($ov \cup w$,$\mathit{\langle\mbox{INSTALL-SEQ}, w, seq, ov\rangle}$)
\STATE \emph{R-multicast}($\{j \colon j \in ov \vee  j \in w\}$,$\mathit{\langle\mbox{INSTALL-SEQ}, w, seq, ov\rangle}$) 
%\STATE $\forall k \in \mathit{ov} \cup w, send \langle\mbox{INSTALL-SEQ}, w, \mathit{seq}, ov\rangle$ to $k$
%
%\renewcommand{\algorithmicrequire}{\textbf{upon receipt of}}
%\vspace{0.5mm}
%\REQUIRE $\mathit{\langle\mbox{INSTALL-SEQ}, w, seq, ov \rangle}$ from $j$ and $\mathit{i \in ov}$
%\STATE $\forall k \in \mathit{ov} \cup w, send \langle\mbox{INSTALL-SEQ}, w, seq, ov\rangle$ to $k$ 
%\COMMENT{\hfill //re-send the message to ensure all processes in $\mathit{ov} \cup w$ will receive the message sent by $j$}

\vspace{0.5mm}
\REQUIRE \emph{R-delivery}($\{j \colon j \in ov \vee  j \in w\}$,$\mathit{\langle\mbox{INSTALL-SEQ}, w, seq, ov\rangle}$) 
% \vspace{1mm}
% \renewcommand{\algorithmicrequire}{\textbf{procedure}}
% \REQUIRE $\mathit{\mathtt{VIEW\_GEN}^{ov}.new\_view\_seq(seq^{ov})}$
% \STATE $w \leftarrow least\_updated(seq^{ov})$
\IF[\hfill //$i$ is member of the previous view in the sequence]{$\mathit{i \in ov}$}
\STATE \textbf{if} $cv \subset w$ \textbf{then} \textit{stop the execution of r/w operations} \COMMENT{\hfill //if $w$ is more up-to-date than $cv$, enqueue r/w operations while the state transfer is done}
\STATE $\forall j \in w, \mathit{send \langle\mbox{STATE-UPDATE}, \langle val, ts\rangle, \mathtt{RECV} \rangle}$ to $j$  
\COMMENT{\hfill //$i$ sends its state to servers in the next view} 
\ENDIF
\IF[\hfill //$w$ is more up-to-date than $cv$ and the system will be reconfigured from $cv$ to $w$]{$cv \subset w$}
% \STATE \textit{stop to execute client r/w operations} \COMMENT{\hfill //enqueue these operations to be executed when the last view in the sequence is installed -- line 20}
\IF[\hfill //$i$ is in the new view]{$\mathit{i \in w}$} 
\STATE  \textbf{wait} for $\mathit{\langle\mbox{STATE-UPDATE},*,*\rangle}$ messages from $ov.q$ servers in $ov$
% \STATE $\langle \mathit{val, ts} \rangle \leftarrow$ read from the quorum of $\mbox{STATE-UPDATE}$ messages received \COMMENT{\hfill //$i$ updates its state...} %in line 8%\COMMENT{update $\langle \mathit{val, ts} \rangle$ if necessary}
\STATE $\langle \mathit{val, ts} \rangle \leftarrow \langle \mathit{val_h, ts_h} \rangle$, pair with highest timestamp among the ones received \COMMENT{\hfill //$i$ updates its state...}
\STATE $\mathtt{RECV} \leftarrow \mathtt{RECV} \cup \{$update requests from $\mbox{STATE-UPDATE}$ messages$\} \setminus \mathit{w.entries}$ 
%\COMMENT{\hfill //... and collects updates requests that were not processed yet...}
\STATE $cv \leftarrow w$ \COMMENT{\hfill //... and updates its current view to $w$}
\STATE \textbf{if} $\mathit{i \not\in ov}$ \textbf{then} \textit{enable operations} \COMMENT{\hfill //$i$ joining the system}
\STATE $\forall j \in ov \setminus cv, \mathit{send \langle\mbox{VIEW-UPDATED}, cv\rangle}$ to $j$ \COMMENT{\hfill //inform servers in $ov \setminus cv$ that they can leave}
% \IF[\hfill //no update proposal for $cv = w$ was made yet]{$\mathtt{SEQ}^{cv} = \emptyset$}
\IF[\hfill //there are views more up-to-date than $cv$ in $\mathit{seq}$...]{$(\exists w' \in \mathit{seq} \colon cv \subset w')$}
 \STATE $\mathit{seq'} \leftarrow \{w' \in seq \colon cv \subset w'\}$ 
  \COMMENT{\hfill //... gather these views...}
%  \STATE $\forall k \in \mathit{cv}, send \langle\mbox{SEQ-VIEW}, \mathtt{SEQ}^{cv}, \emptyset \rangle$ to $k$
%   \COMMENT{\hfill //... and propose them (going back to Algorithm \ref{alg_recon_vivo})}
\STATE $\mathcal{G}^{cv}_i.\mathit{gen\_view}(seq')$ 
\COMMENT{\hfill //... and propose them (going back to Algorithm \ref{alg_live})}
\ELSE
 \STATE \emph{resume the execution of r/w operations} in $cv = w$ and start a timer for $cv$ \COMMENT{\hfill //\textbf{$w$ is installed}}
% \STATE start timer for $\mathit{cv}$ \COMMENT{\hfill //for triggering the next reconfiguration}
\ENDIF
\ELSE[\hfill //$i$ is leaving the system]
\STATE \textit{disable operations}
\STATE \textbf{wait} for $\langle\mbox{VIEW-UPDATED}, w\rangle$ messages from $w.q$ servers in $w$ and then \textbf{halt}
% \COMMENT{\hfill //\textbf{$w$ is installed}}
\ENDIF
% \ENDIF
\ENDIF

\end{algorithmic}
}
  \vspace{-1mm}
\end{algorithm*}
%   \vspace{-2mm}
%%%%%%%%%%%%%%%%%%%%%%%%%%%%%%%%%%%%%%%

\paragraph*{The protocol}

Algorithm \ref{alg_recon} (lines 8-28) presents the core of the \textsc{FreeStore} reconfiguration protocol.
This algorithm uses an auxiliary function $\mathit{least\_updated}$ to obtain the least updated view in a sequence of views (i.e., the one that is contained in all other views of the sequence) and two local variables: the aforementioned $\mathtt{RECV}$ -- used to store pending reconfiguration requests -- and $cv$ -- the current view of the server (initially $v_0$).

When the view generator associated with some view $ov$ reports\footnote{We use $ov$ instead of $cv$ because view generators associated with \emph{old views} 
($ov \subseteq cv$) still active can generate new sequences.} the generation of a sequence of views $seq$, the server obtains the least 
updated view $w$ of $seq$ and proposes this sequence for updating $ov$ through an $\mbox{INSTALL-SEQ}$ message sent to the servers of both, $ov$ and $w$.
To ensure all correct servers in this group will process the reconfiguration, we employ a \emph{reliable multicast} primitive~\cite{Had94} (lines 8-9), 
that can be efficiently implemented in asynchronous systems with a message retransmission at the receivers before its delivery (i.e., in one communication step)~\cite{Cha96,Had94}.

The current view is updated through the execution of lines 10-28.
First, if the server is a member of the view being updated $ov$, it must send its state (usually, the register's value and timestamp) to the servers in the new view to be installed (lines 10-12).
However, if the server will be updating its current view (i.e., if $w$ is more up-to-date than $cv$) it first needs to stop executing client r/w operations (line 11) and enqueue these operations to be executed when the most up-to-date view in the sequence is installed (line 25, as discussed bellow).
%Besides the storage state of a server, a $\mbox{STATE-UPDATE}$ message also carry the set of pending updates ($\mathtt{RECV}$) of this server to ensure that, as with writes, no update request is lost during reconfigurations.
% 
% 
% 
% 
% %%%%%%%%%%%%%%%%%%%%%%%%%%%%%%%%%%%%%%%
% \begin{figure*}[!ht]
% \begin{center}
% \includegraphics[width=0.8\textwidth]{figs/freestore_rec_new}
% \end{center}
%  \vspace{-0.5cm}
% \caption{Processes always converge to a single view, even in the presence of conflicting proposals.}
% \label{fig:exec}
%  \vspace{-0.3cm}
% \end{figure*}
% %%%%%%%%%%%%%%%%%%%%%%%%%%%%%%%%%%%%%%%
% 
A server will update its current view only if the least updated view $w$ of the proposed sequence is more up-to-date than its current view $cv$ (line 13).
If this is the case, either (1) server $i$ will be in the next view (lines 14-25) or (2) not (line 26-28).

\begin{itemize}

\item \textsc{Case 1} [$i$ will be in the next view (it may be joining the system)]:
If the server will be in the next view $w$, it first waits for the state from a quorum of servers from the previous view $ov$ and then defines the current value 
and timestamp of the register (lines 15-16), similarly to what is done in the 1st phase of a read operation (see Section~\ref{rw_prt}).
% and update its pending reconfigurations set $\mathtt{RECV}$ (lines 15-17).
After ensuring that its state is updated, the server updates $cv$ to $w$ and, if it is joining the system, it enables the processing of r/w operations (which will be queued until line 25 is executed).
Furthermore, the server informs leaving servers that its current view was updated (line 20).
The final step of the reconfiguration procedure is the verification if the new view will be installed or not (i.e., an auxiliary view).
If $cv=w$ is not the most up-to-date view of the generated sequence $seq$, a new sequence with the views more up-to-date than $cv$ will be proposed for updating it (lines 21-23).
Otherwise, $cv$ is installed and server $i$ resumes processing r/w operations (lines 24-25).

\item \textsc{Case 2} [$i$ is leaving the system]: A server that is leaving the system only halts after ensuring that the view $w$ to which it sent its state was started in a quorum of servers (lines 27-28).
%This is required for ensuring the reconfiguration liveness.

\end{itemize}

Although the algorithm (in some sense) restarts itself in line 23, the fact that the number of reconfiguration requests is finite (Assumption~\ref{finite_rec}) makes the reconfiguration procedure eventually terminate.
Furthermore, since all sequences generated by $\mathcal{L}$ are compatible, when the reconfiguration terminates in all correct servers, they will have installed the same view with all requested updates. 
% Appendix~\ref{ap:exec} presents an example of execution of the protocol.

\paragraph*{Optimization}
It is possible to reduce one communication step in the \textsc{FreeStore} reconfiguration protocol (Algorithm~\ref{alg_recon}) when using $\mathcal{L}$ (Algorithm \ref{alg_live}). 
The idea is to use the $\mbox{SEQ\_CONV}$ messages of $\mathcal{L}$ to inform the servers about a new sequence to be installed, making thus the $\mbox{INSTALL\_SEQ}$ message unnecessary.

\paragraph*{An Execution Example}

Figure \ref{fig:exec} illustrates an execution of \textsc{FreeStore} reconfiguration with $\mathcal{L}$ that makes a set of 
  servers converge to a single view, even in the presence of conflicting proposals.
  Let $v_0=\{1,2,3\}$ be the initial view of a system using our algorithm with the $\mathcal{L}$ view generator.
  In this example, \ding{192} servers $1,2,3$ receive the $join$ request from process $4$ before they start their view generators, while the join request from server $5$ was received only by server $1$.
  Consequently, the view generators associated with $v_0$ ($L(v_0)$) are started at servers $2,3$ with a sequence containing a view $v_1=\{1,2,3,4\}$ and at server $1$ with a sequence containing a view $v_2=\{1,2,3,4,5\}$.
  The Weak Accuracy of $L(v_0)$ allows \ding{193} servers $2,3$ to generate a sequence $seq_1: v_1$  and later servers $1,2,3$ may generate the sequence $seq_2: v_1 \rightarrow v_2$ also for updating $v_0$.%, i.e. with $seq_1 \subset seq_2$.

   %%%%%%%%%%%%%%%%%%%%%%%%%%%%%%%%%%%%%%%%
\begin{figure*}[!htb]
\begin{center}
\includegraphics[width=0.65\textwidth]{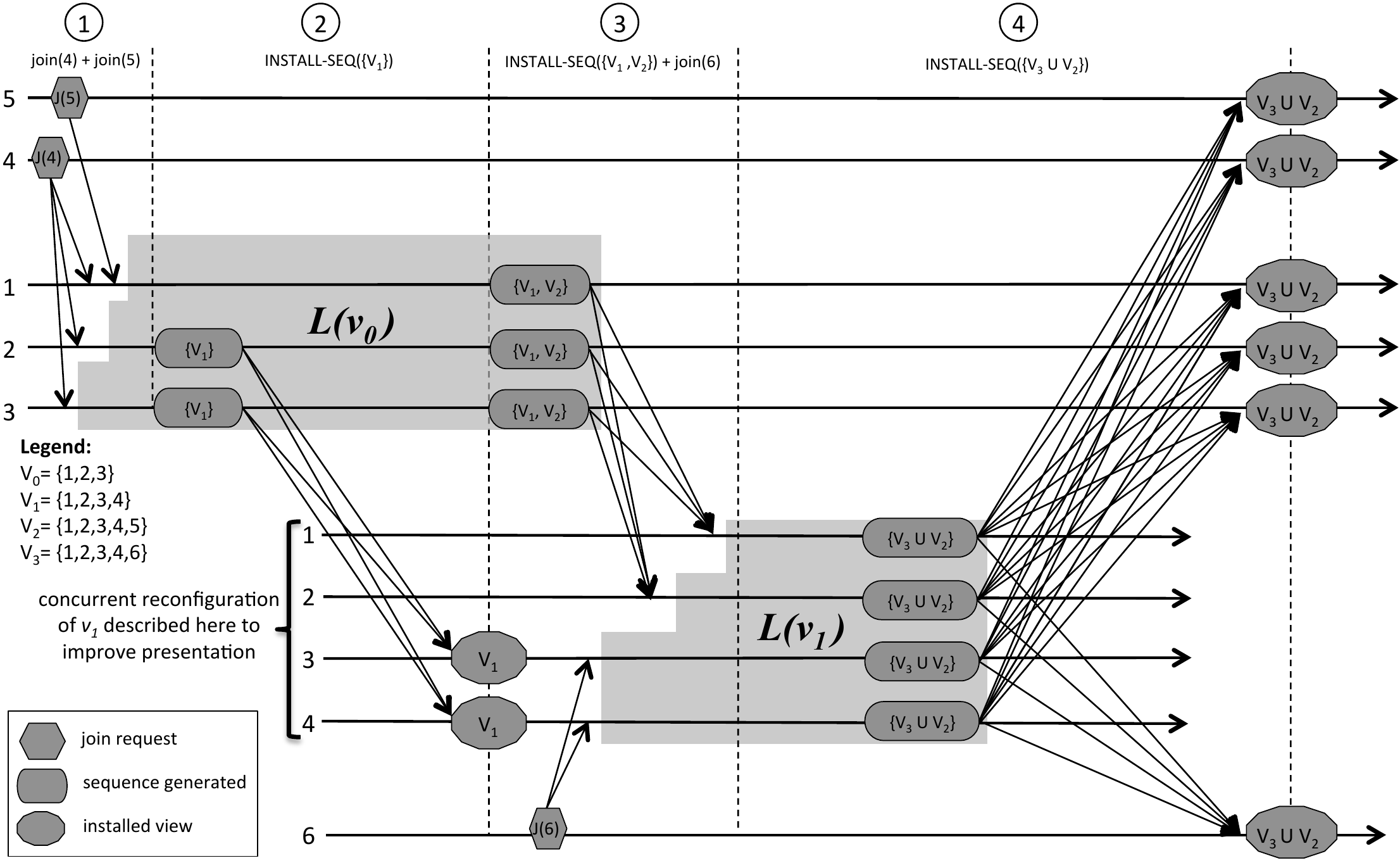}
\end{center}
% \vspace{-0.5cm}
\caption{Processes always converge to a single view, even in the presence of conflicting proposals.}
\label{fig:exec}
\end{figure*}
%%%%%%%%%%%%%%%%%%%%%%%%%%%%%%%%%%%%%%%%

In this execution, \ding{193} servers $3,4$ receive $\mbox{INSTALL-SEQ}$ messages for $seq_1$ and \ding{194} install $v_1$ before the reception of $\mbox{INSTALL-SEQ}$ messages for $seq_2$.
  Additionally, \ding{194} servers $3,4$ start their view generators $L(v_1)$ for updating $v_1$ to add server $6$, which sends a join request to them.
  More specifically, they propose a sequence $seq_3: v_3$ with $v_3 = \{1,2,3,4,6\}$ (line 7). 
  Concurrently, servers $1,2$ receive the $\mbox{INSTALL-SEQ}$ messages for $seq_2$ before receiving these messages for $seq_1$.
  Consequently, \ding{194} $v_1$ will be an auxiliary view at servers $1,2$ and they will start $L(v_1)$ with a sequence containing $seq_4: v_2$ (line 23).
  
  In this scenario, servers $1,2,3,4$ (from $v_1$) have to handle conflicting sequence proposals for updating $v_1$, since neither $v_3 \subset v_2$ nor $v_2 \subset v_3$.
  Depending on the schedule of messages received during the convergence of $L(v_1)$ (i.e., the composition of the first quorum of $\mbox{SEQ-VIEW}$ messages received by processes in $v_1$), the new sequence could be either (1)~$v_2$, (2) $v_3$, (3)~$v_2 \rightarrow (v_3 \cup v_2)$, (4)~$v_3 \rightarrow (v_3 \cup v_2)$ or (5)~$v_3 \cup v_2$.
  However, the Weak Accuracy property of $L(v_1)$ ensures that servers will not install both $v_3$ and $v_2$ (the conflicted views).
  Since neither $\{1,2\}$ nor $\{3,4\}$ form a quorum of $v_1$ (which requires three servers), they must deal with the conflicting proposals before converging to some sequence (lines 4-8 of Algorithm~\ref{alg_live}).  
  Consequently, the only allowed sequence is (5).
  Thus, \ding{195} $L(v_1)$ generates a sequence $seq_5: v_3 \cup v_2$ and all servers install the view $v_3 \cup v_2 = \{1,2,3,4,5,6\}$.

\paragraph*{Correctness (full proof in Section \ref{sec:basic})}
Algorithm \ref{alg_recon} ensures that \textit{an unique sequence of views is installed in the system} by the following: 
(1) if a view $w$ is installed, any previously installed view $w' \subset w$ is uninstalled and will not be installed anymore (lines 11, 18 and 25); (2) consequently, no view 
more up-to-date than $w$ is installed and, thus, the installed views form an unique sequence. Moreover, by Assumption~\ref{finite_rec} the reconfiguration procedure always terminate by installing a final view $v_{final}$. 
The \emph{Storage Safety} and \emph{Storage Liveness} properties are discussed in Section \ref{rw_prt}.
The remaining properties of Definition \ref{def_freestore} are ensured as follows. 
\textit{Reconfiguration Join/Leave Safety} are trivially ensured by the fact that only a server $i$ sends the reconfiguration request $\langle +,i \rangle$/$\langle -,i \rangle$ (lines 1-4).
\textit{Reconfiguration Join/Leave Liveness} are ensured by the fact that if an update request from a server $i$ is stored in $\texttt{RECV}$ of a quorum, then it will be processed in the next reconfiguration.
This happens because a quorum with the same proposal is required to generate some view sequence (Algorithm \ref{alg_live}), and this quorum intersects with the servers that acknowledged request of $i$ in at least one correct server. 
%if $i$ does not receive the acknowledgment from a quorum, then $i$ receives an up-to-date view and issues its update request again in this view.
Moreover, update requests received during a reconfiguration are sent to the next view (lines 12-17).

\subsection{Reconfiguration using $\mathcal{P}$}
\label{rec_ger_perfeito}

If $\mathcal{P}$ is used with the \textsc{FreeStore} reconfiguration protocol (Algorithm \ref{alg_recon}), all generators will generate the same sequence of views (Strong Accuracy) with a single view $w$ (obtained in line 8).
This will lead the system directly from its current view $cv$ to $w$ (lines 22-23 will never be executed).

\section{Read and Write Protocols}
\label{rw_prt}

This section discusses how a static storage protocol can be adapted to dynamic systems by using \textsc{FreeStore} reconfigurations. 
Since reconfigurations are decoupled from r/w protocols, they are very similar to their static versions.
In a nutshell, there are two main requirements for using our reconfiguration protocol.
First, each process (client or server) needs to handle a current view variable $cv$ that stores the most up-to-date view it knows.
All r/w protocol messages need to carry $cv$ and clients update it as soon as they discover that there is a more recent view installed in the system. 
The servers, on the other hand, reject any operation issued to an old view, and reply their current view to the issuing client, which updates its $cv$.
% The client then updates its $cv$ and restarts the phase of the operation it is executing. % (see below).
The client restarts the phase of the operation it is executing if it receives a view different from the one that it handles. % (see below).

The second requirement is that, before accessing the system, a client must obtain the system's current view.
This can be done by making servers put the current view in a directory service~\cite{Agu10} or making the client flood the network asking for it.
Notice this is an intrinsic problem for any dynamic system and similar assumptions are required in previous reconfiguration protocols~\cite{Gil10,Mar04,Agu11,Jen15,Gaf15}.
\begin{algorithm*}
% \vspace{-2mm}
\floatname{algorithm}{Algorithm} \caption{Algorithm executed at client \textit{c}}
\label{alg-client}
\begin{multicols}{2}
%  \begin{spacing}{0.7}
\scriptsize{
\begin{algorithmic}[1]
\renewcommand{\algorithmicrequire}{\textbf{procedure}}

\REQUIRE $\mathit{write(value)}$\\
\COMMENT{\textbf{Phase 1}} \\

\COMMENT{ABD Phase 1} \\

\fbox{\parbox{6.5cm}{
\STATE $\mathit{send \langle \mbox{READ\_TS} \rangle}$  $to$ $each$ $i \in cv$
\REPEAT
\STATE $\textbf{wait}$ $reply$ $\mathit{\langle\mbox{READ\_TS\_REP},ts_i, v_i\rangle}$ $from$ $i$ $\in cv$

\STATE $\mathit{\mathtt{RTS}_i \leftarrow \langle\mbox{READ\_TS\_REP},ts_i, v_i\rangle}$
\UNTIL{$\mathit{|\mathtt{RTS}| \geq cv.q	}$}
\STATE $\mathit{ts \leftarrow succ(max(\mathtt{RTS}),c)}$
}}

\vspace{1mm}
\COMMENT{FreeStore Phase 1 (updated view verification)} \\

\fbox{\parbox{6.5cm}{
\IF{$\mathit{\exists \langle\mbox{READ\_TS\_REP},ts_j, v_j\rangle \in \mathtt{RTS} \colon cv \neq v_j}$}
\STATE \textbf{if} $\mathit{cv \subset v_j}$ \textbf{then} $\mathit{cv \leftarrow v_j}$
\STATE $restart$ $phase$ $1$
\ENDIF
}}

% \STATE $\mathit{ts_{max} \leftarrow}$ 

\vspace{4mm} 
\COMMENT{\textbf{Phase 2}} \\
\COMMENT{ABD Phase 2} \\

\fbox{\parbox{6.5cm}{
\STATE $\mathit{send \langle\mbox{WRITE},value, ts\rangle}$  $to$ $each$ $i \in cv$
\REPEAT
\STATE $\textbf{wait}$ $reply$ $\mathit{\langle\mbox{WRITE\_REP}, v_i\rangle}$ $from$ $i \in cv$
\STATE $\mathit{\mathtt{WROTE}_i \leftarrow \langle\mbox{WRITE\_REP}, v_i\rangle}$
\UNTIL{$\mathit{|\mathtt{WROTE}| \geq cv.q}$}
}}

\vspace{1mm}
\COMMENT{FreeStore Phase 2 (updated view verification)} \\

\fbox{\parbox{6.5cm}{
\IF{$\mathit{\exists \langle\mbox{WRITE\_REP}, v_j\rangle \in \mathtt{WROTE} \colon cv \neq v_j}$}
% \IF{$\mathit{cv \subset v_i}$}
% \STATE $\mathit{cv \leftarrow v_i}$
% \ENDIF
\STATE \textbf{if} $\mathit{cv \subset v_j}$ \textbf{then} $\mathit{cv \leftarrow v_j}$
\STATE $restart$ $phase$ $2$
\ENDIF
}}

\vspace{6mm}

\REQUIRE $\mathit{read()}$\\
\COMMENT{\textbf{Phase 1}} \\
\COMMENT{ABD Phase 1} \\
\fbox{\parbox{7cm}{
\STATE $\mathit{send \langle \mbox{READ} \rangle}$ $to$ $each$ $i \in cv$
\REPEAT
\STATE {$\textbf{wait}$ $reply$ $\mathit{\langle\mbox{READ\_REP}, \langle val, ts\rangle_i , cv_i\rangle}$ $from$ $i \in cv$}
\STATE $\mathit{\mathtt{READ}_i \leftarrow \langle\mbox{READ\_REP}, \langle val, ts\rangle_i , cv_i\rangle}$
\UNTIL {$\mathit{|\mathtt{READ}| \geq cv.q}$}
\STATE $\mathit{\langle val_h, ts_h\rangle \leftarrow maxTS(\mathtt{READ})}$ 
}}

\vspace{1mm}
\COMMENT{FreeStore Phase 1 (updated view verification)} \\
\fbox{\parbox{7cm}{

\IF{$\mathit{\exists \langle\mbox{READ\_REP}, \langle val, ts\rangle_j , cv_j\rangle \in \mathtt{READ} \colon cv \neq v_j}$}
% \IF{$\mathit{cv \subset v_i}$}
% \STATE $\mathit{cv \leftarrow v_i}$
% \ENDIF
\STATE \textbf{if} $\mathit{cv \subset v_j}$ \textbf{then} $\mathit{cv \leftarrow v_j}$
\STATE $restart$ $phase$ $1$
\ENDIF
% \UNTIL{$\mathit{|\mathtt{READ}| \geq cv.q}$}
}}

 \vspace{2mm}
 \COMMENT{\textbf{Phase 2 (write-back phase)}} \\
 \COMMENT{ABD Phase 2} \\
 \fbox{\parbox{7cm}{
 \IF{$\mathit{\exists \langle\mbox{READ\_REP}, \langle *, ts\rangle_* , *\rangle \in \mathtt{READ} \colon ts \neq ts_h}$}
 \STATE $\mathit{send \langle\mbox{WRITE}, val_h, ts_h \rangle}$ $to$ $each$ $i \in cv$
 \REPEAT
 \STATE {$\textbf{wait}$ $reply$ $\mathit{\langle\mbox{WRITE\_REP}, cv_i\rangle}$ $from$ $i \in cv$}
 \STATE $\mathit{\mathtt{WROTE}_i \leftarrow \langle\mbox{WRITE\_REP}, cv_i\rangle}$
 \UNTIL{$\mathit{|\mathtt{WROTE}| \geq cv.q}$}
 \ENDIF
 }}

\vspace{1mm}

 \COMMENT{FreeStore Phase 2 (updated view verification)} \\
 
 \fbox{\parbox{7cm}{
 \IF{$\mathit{\exists \langle\mbox{WRITE\_REP}, v_j\rangle \in \mathtt{WROTE} \colon cv \neq v_j}$}
 % \IF{$\mathit{cv \subset v_i}$}
 % \STATE $\mathit{cv \leftarrow v_i}$
 % \ENDIF
 \STATE \textbf{if} $\mathit{cv \subset v_j}$ \textbf{then} $\mathit{cv \leftarrow v_j}$
 \STATE $restart$ $phase$ $2$
 \ENDIF
 }}

\RETURN $\mathit{val_h}$

\end{algorithmic}

}
\end{multicols}
\end{algorithm*}

% \vspace{-3mm}  
\begin{algorithm}
% \vspace{-2mm}  
\floatname{algorithm}{Algorithm} \caption{Algorithm executed at server \textit{i}.}
\label{alg-server}
% \begin{multicols}{2}
% \begin{spacing}{0.7}
\scriptsize{
\begin{algorithmic}[1]

\vspace{2mm}
 \COMMENT{ABD server algorithm} 
\newline
 \fbox{\parbox{7.5cm}{
\renewcommand{\algorithmicrequire}{\textbf{   \hspace{4mm }Upon receipt of}}

\REQUIRE $\mathit{\langle\mbox{READ\_TS} \rangle}$ $from$ $client$ $c$
% \IF{$\mathit{verify(view_c) \wedge currV \subset view_c}$}
% \IF{$\mathit{i \not\in view_c}$}
% \STATE $\mathit{halt}$
% \ENDIF
% \STATE $\mathit{\textbf{checkServerView}(cv_c)}$
% \ENDIF
\STATE $\mathit{send \langle\mbox{READ\_TS\_REP}, ts, cv\rangle}$ $to$ $c$

\vspace{1mm}

\REQUIRE $\mathit{\langle\mbox{WRITE},val_w, ts_w \rangle}$ $from$ $client$ $c$
% \STATE $\mathit{\textbf{checkServerView}(cv_c)}$
% \IF{$\mathit{isValid(P_{new}) \wedge P_{new}.hash = h(value) \wedge}$ \\
% \hfill$\mathit{isValid(cv_c) \wedge (cv = cv_c) \wedge}$ \\
%  \hfill $\mathit{(cv.rd\_info = cv_c.rd\_info)}$}
% \IF{$\mathit{isValidPC(P_{new}) \wedge P_{new}.hash = h(value)}$}
\IF{$\mathit{ts_w > ts}$}
\STATE $\mathit{val \leftarrow val_w}$
\STATE $\mathit{ts \leftarrow ts_w}$
\ENDIF
% \STATE $\mathit{\langle a_i,v_i\rangle \leftarrow \mathit{Th\_sign}(SK_i,VK_i,VK,P_{new}.ts)}$
\STATE $\mathit{send \langle\mbox{WRITE\_REP}, cv\rangle}$ $to$ $c$
% \ELSE[$cv_c$ is old or $c$ is faulty] 
% \STATE $\mathit{send \langle\mbox{WRITE\_REP},\perp,\perp, currV\rangle}$ $to$ $c$
% \ENDIF

\vspace{1mm}

\REQUIRE $\mathit{\langle\mbox{READ}\rangle}$ $from$ $client$ $c$
% \IF{$\mathit{view_p \neq \perp}$}
% \STATE $\mathit{\textbf{checkServerView}(cv_c)}$
% \ENDIF
\STATE $\mathit{send \langle\mbox{READ\_REP}, \langle val,ts \rangle, cv\rangle}$ $to$ $c$

 }}
\end{algorithmic}

}
% \end{spacing}

% \end{multicols}
% \vspace{-5mm}  
\end{algorithm}

In this paper we extend the classical ABD algorithm~\cite{Att95} for supporting multiple writers and to work with the \textsc{FreeStore} reconfiguration. 
In the following we highlight the main aspects of these r/w protocols (algorithms~\ref{alg-client} and \ref{alg-server}).
In the algorithms, we highlight the ABD code as well as the code related with \textsc{FreeStore}. 

The protocols to read and write from the dynamic distributed storage work in phases.
Each phase corresponds to an access to a quorum of servers in $cv$.
The \emph{read protocol} works as follows: 

\begin{itemize}

\item (\textsc{1st Phase}) a reader client requests a set of tuples $\langle val,ts \rangle$ from a quorum of servers in $cv$ 
($val$ is the value that the server stores and $ts$ is its associated timestamp) and selects the one with highest timestamp $\langle val_h,ts_h \rangle$;  
the operation ends and returns $val_h$ if all returned pairs are equal, which happens in executions without write contention or failures; 

\item (\textsc{2nd Phase}) otherwise, the reader client performs an additional \textit{write-back phase} in the system and waits for confirmations from a quorum of servers in $cv$ before returning $val_h$.

\end{itemize}

The \emph{write protocol} works in a similar way:

\begin{itemize}
\item (\textsc{1st Phase}) a writer client obtains a set of timestamps from a quorum of servers in $cv$ and chooses the highest, $ts_h$; 
the timestamp to be written $ts$ is defined by incrementing $ts_h$ and concatenating the writer id in its lowest bits; 

\item (\textsc{2nd Phase}) the writer sends a tuple $\langle val,ts \rangle$ to the servers of $cv$, writing $val$ with timestamp $ts$, and waits for confirmations from a quorum.
\end{itemize}

% 
% The proposed decoupling of r/w protocols from the reconfiguration makes it possible to run r/w operations in parallel with reconfigurations and 
% makes it easy to adapt other static fault-tolerant register implementation to dynamic environments.
% This happens because, differently from previous approaches~\cite{Gil10,Agu11,Jen15,Gaf15}, where r/w operations may be executed in multiple views of the system, in \textsc{FreeStore} a client 
% execute these operations only in the most up-to-date view installed in the system. 
% % To enable this, it is required that r/w operations to be blocked during the state transfer between views (Algorithm \ref{alg_recon}).

The proposed decoupling of r/w protocols from the reconfiguration makes it easy to adapt other static fault-tolerant register implementation to dynamic environments 
and makes it possible to run r/w operations in parallel with reconfigurations. This happens because, differently from previous approaches~\cite{Gil10,Agu11,Jen15,Gaf15}, where r/w operations may 
be executed in multiple views of the system, in \textsc{FreeStore} a client 
executes these operations only in the most up-to-date view installed in the system, which is received directly from the servers. 
To enable this, it is required that r/w operations to be blocked during the state transfer between views (Algorithm \ref{alg_recon}).

% \vspace{-2mm}
% \paragraph{Sketch of the correctness proof}
\paragraph*{Correctness (full proof in Section \ref{sec:free})}

The above algorithms implement atomic storage in the absence of reconfigurations~\cite{Att95}.
When integrated with \textsc{FreeStore}, they also have to satisfy the \textit{Storage Safety} and  \textit{Storage Liveness} properties of Definition~\ref{def_freestore}.
We start by discussing the \emph{Storage Liveness} property, which follows directly from the termination of reconfigurations.
%
%\begin{cor} \label{cor_termination}
%Clients' r/w operations always terminate.
%\end{cor}
%
%\vspace{-1mm}
%\begin{proof} (Sketch) 
As discussed before, clients' r/w operations are concluded only if they access a quorum of servers using the same view as the client, otherwise they are restarted.
% From Lemma \ref{lemma_rec_terminacao}, the system reconfiguration always terminate by installing some final view.
By Assumption \ref{finite_rec}, the system reconfiguration always terminate by installing some final view.
Consequently, a client will restart phases of its operations a finite number of times until this final view is installed in a quorum.
%\end{proof}
% 
\emph{Storage Safety} comes directly from three facts: (1) a r/w operation can only be executed in a single \emph{installed} view (i.e., all servers in the quorum of 
the operation have the same installed current view), (2) all installed views form a unique sequence%(Lemma~\ref{lemma_rec_sequencia})
, and (3) any operation executed in a view $v$ will still be ``in effect'' when a more up-to-date view $w$ is installed.
More precisely, assume $\langle val,ts \rangle$ is the last value read or written in $v$, and thus it was stored in $v.q$ servers from $v$.
During a reconfiguration from $v$ to $w$, r/w operations are disabled until all servers of $v$ send $\langle val,ts \rangle$ to the servers of $w$ (line 12), which terminate the reconfiguration only after receiving the state from a quorum of servers of $v$.
Consequently, all servers who reconfigure to $w$ will have $\langle val,ts \rangle$ as its register's state (lines 15-16).
This ensures that any operation executed in $w$ will not miss the operations executed in~$v$.

\section{\textsc{FreeStore} Correctness Proofs}
\label{sec:free_prop}

\subsection{Basic Properties}
\label{sec:basic}

The following lemmata show a sketch of the proofs that Algorithm \ref{alg_recon} with $\mathcal{L}$ installs an unique sequence of views in the systems and that the reconfiguration procedure always terminate.\footnote{We consider $\mathcal{L}$ because it may generate more than just one sequence of views (Weak Accuracy). These proofs are also valid for Algorithm~\ref{alg_recon} with $\mathcal{P}$ since Strong Accuracy is a special case of Weak Accuracy (see Section~\ref{geradores}).}

Recall that a view $v$ is \emph{installed} in the system if some correct server $i \in v$ considers $v$ as its \emph{current view} (var $cv$ in line 18 - Algorithm~\ref{alg_recon}) and enable r/w operations in this view (line 19). 
When this happens, we say that the previous view, which was installed before $v$, was \emph{uninstalled} from the system (Section~\ref{dyn_def}).

\begin{lem} \label{lemma_rec_vivo_vAuxiliar}\vspace{-0mm}
Let $v$ to be a view installed in the system and $\mathcal{G}^{v}$ its associated view generator. 
If $\mathcal{G}^{v}$ generates $seq: v_{1} \rightarrow ... \rightarrow v_{k} \rightarrow w$, then only $w$ will be installed in the system. 
\end{lem}
 
\vspace{-1mm}
\begin{proof}
Given the sequence $seq: v_{1} \rightarrow ... \rightarrow v_{k} \rightarrow w$ generated by $\mathcal{G}^{v}$,
we prove by induction on $k$ that views $v_1$ to $v_k$ will not be installed.
The base case ($k = 1$) is trivial.
Servers in $v_1$ receive the state from $v$ (line 15) and do not install this view since the condition of line 21 is satisfied.
Therefore, each $i \in v_1$ starts its view generator $\mathcal{G}^{v_1}_i$ by proposing the same sequence $seq_{v_1}: seq \setminus \{v_1\}$ (lines~22-23), which is eventually generated by $\mathcal{G}^{v_1}$ since all proposals are equal to $seq_{v_1}$.
In the induction step, we have to show if the claim is valid for $k-1$, it will be valid for $k$.
When processing $seq_{v_{k-1}}: v_k \rightarrow w$, each server in $v_k$ receives the state from $v_{k-1}$ and will not install $v_k$
since the condition of line 21 will be true. Consequently, views $v_1$ to $v_k$ will not be installed. Now, let us show that $w$ is installed.
Since servers in $v_k$ receive the sequence $seq_{v_{k-1}}$, the view generators $\mathcal{G}^{v_k}$ will be started with a
proposal $seq_{v_k}: seq_{v_{k-1}} \setminus \{v_k\} = \{w\}$, which eventually will be generated.
During the processing of $seq_{v_k} = w$ the predicate of line 21 will be false and $w$ will be installed (line 25).
Thus, only $w$ from $seq$ is installed in the system.
% $\strut\hfill \Box$
\end{proof}

% \begin{lem} \label{lemma_rec_vivo_installW}
% When a view $w$ is installed in the system, any previous view $w'$, such that $w' \subset w$, is not active anymore.
% \end{lem}
%  
% \begin{proof}
% Consider that $w$ is installed in the system.
% In this case, the view generators $vg_v$ associated with some view $v$ have defined a sequence $seq$ for its update, such that $w$ is the most up-to-date view.
% Then, for $w$ to be installed, (1) view generators of all views in $seq \setminus w$ have defined sequences for view updates until arriving to a view $w_{aux}$, immediately before $w$ in $seq$ and (2) $w_{aux}$ has defined a sequence only with $w$, in such a way that $w$ is installed (line 30). 
% In this case, all the other views in $seq$ are deactivated (lines 17 and 23).
% Now, consider that a previous sequence $seq'$ generated by $vg_v$ led to the installation of another view $w'$. From the weak accuracy of $vg_v$, $seq' \subset seq$ and, then, $w' \in seq$. Since $w' \in seq$, $w'$ is deactivated on the installation of $w$. 
% From the weak accuracy of $vg_v$ and the Lemma~\ref{lemma_rec_vivo_vAuxiliar}, any subsequent sequence$seq''$ generated by $vg_v$ is such that $seq \subset seq''$, then $w'$ will never be installed again. 
% By induction, any view $w'$, such that $w' \subset w$, is deactivated when $w$ is installed and will not be installed anymore. 
% \end{proof}

\begin{lem} \label{lemma_rec_vivo_installW}
If a view $w$ is installed in the system, any previously installed view $w' \subset w$ was uninstalled and will not be installed anymore.
\end{lem}

\vspace{-1mm} 
\begin{proof}
Consider that $w$ is installed in the system.
In this case, the view generators $\mathcal{G}^{v}$ associated with some view $v$ have defined a sequence $seq: v_1 \rightarrow ... \rightarrow v_k \rightarrow w$ for installing $w$ (Lemma \ref{lemma_rec_vivo_vAuxiliar}).
Consequently, all other views in $seq$ are uninstalled (line 18). 
Now, consider that a previous sequence $seq': v_1 \rightarrow ... \rightarrow v_y \rightarrow w'$, $y < k$, generated by $\mathcal{G}^{v}$ led to the installation of another view $w'$.
From the Weak Accuracy property of $\mathcal{G}^{v}$, $seq' \subseteq seq$ and $w' \in seq$.
Consequently, the fact that $w' \in seq$ implies that $w'$ is uninstalled on the installation of $w$.
From the Weak Accuracy of $\mathcal{G}^{v}$, any subsequent sequence $seq''$ generated by  $\mathcal{G}^{v}$ wold have $seq \subseteq seq''$.
Since the algorithm does not install outdated views (line 13), $w'$ will never be installed after $w$.
%This shows that any view $w'$ uninstalled when $w$ is installed will not be installed anymore. 
% $\strut\hfill \Box$
\end{proof}

\begin{lem} \label{lemma_rec_vivo_desativaW}
If a view $w$ is installed in the system, no other more up-to-date view $w'$ such that $w \subset w'$ was installed in the system.
\end{lem}

\vspace{-1mm}
\begin{proof}
Assume there is a view $w'$, such that $w \subset w'$, and that this view is installed in the system. 
From Lemma~\ref{lemma_rec_vivo_installW}, $w$ should not be installed, which implies a contradiction. Consequently, $w'$ is not installed 
in the system.
% $\strut\hfill \Box$
\end{proof}

\begin{lem} \label{lemma_rec_sequencia}
The views installed in the system form an unique sequence.
\end{lem}

\vspace{-1mm}
\begin{proof}
Consider that a view $w_1$ is installed in the system.
From Lemma \ref{lemma_rec_vivo_installW}, any previous outdated view is uninstalled from the system.
From Lemma \ref{lemma_rec_vivo_desativaW}, no view more up-to-date than $w_1$ is installed.
Now, consider that a reconfiguration makes the system reconfigure from $w_1$ to another view $w_2$.
It means that $w_1 \subset w_2$ (line 13), that $w_1$ was uninstalled from the system (Lemma \ref{lemma_rec_vivo_installW}) and that no view more up-to-date than $w_2$ was 
installed (Lemma \ref{lemma_rec_vivo_desativaW}).
The same argument can be applied inductively to any number of views, and thus it is possible to see that the views installed in the system form 
an unique sequence $w_1 \rightarrow w_2 \rightarrow ... \rightarrow w_{k-1} \rightarrow w_k$.
% From Lemma~\ref{lemma_rec_vivo_vAuxiliar}, only the most up-to-date view $w$ of a sequence 
% $seq: v_{1} \rightarrow ... \rightarrow v_{k} \rightarrow w$ generated for system update will be installed. 
% The view generators of the auxiliary views in $seq$ do not introduce new views in their generated sequences 
% since they are started at line 23 receiving as input only views in $seq$. Consequently, only view generators associated 
% with views that are installed in the system may add new views in the sequences since they are started at line 7 receiving 
% by input the updates requests received. 
% 
% Furthermore, when a view $w$ is installed, all previous views $w' \subset w$ have already been uninstalled (Lemma~\ref{lemma_rec_vivo_installW}) and no view more up-to-date than $w$ has been installed before (Lemma~\ref{lemma_rec_vivo_desativaW}).
% This implies that the view generators associated with $w$ may only generate sequences with new views, 
% since no installed view is going to not belong to this sequence, once there is no such view (???). 
% Thus, the views installed in the system form an unique sequence.
% $\strut\hfill \Box$
\end{proof}

\begin{lem} \label{lemma_rec_terminacao}
The reconfiguration procedure always terminate.
\end{lem}

\vspace{-1mm}
\begin{proof}
Consider the system starts from an initial view $v_0$. 
The system reconfiguration starts by initializing view generators $\mathcal{G}^{v_0}$ (line 7), that will define some sequence for updating $v_0$ (lines 8-9).
Lemma \ref{lemma_num_seq} states that the number of sequences generated by $\mathcal{G}^{v_0}$ is finite.
%and, as the number of sequences generated by the view generators associated with the views in these sequences also is finite (Lemma \ref{lemma_num_seq}), 
The reconfiguration started for $v_0$ is going to terminate when the last view $w$ of the last sequence is installed.
In the same way, a reconfiguration started by $\mathcal{G}^{w}$ makes the system reconfigure from 
$w$ to other updated view and finishes.
However, by Assumption \ref{finite_rec}, there will be some view $v_{final}$ such that $\mathcal{G}^{v_{final}}$ will never be started since $\mathtt{RECV} = \emptyset$ in all servers of $v_{final}$ (line 7).
At this point the reconfiguration procedure terminates.
% $\strut\hfill \Box$
\end{proof}

\subsection{\textsc{FreeStore} Properties}
\label{sec:free}

The following lemmata and theorems prove that Algorithm~\ref{alg_recon} and the ABD adaptation for \textsc{FreeStore} presented in Algorithms \ref{alg-client} and \ref{alg-server} 
satisfy the properties of Definition~\ref{def_freestore}.

\subsubsection{Liveness}

Let us now show a sketch of proof that \textsc{FreeStore} satisfies the liveness properties. 
Consider that $V(t)$ is the most up-to-date view installed in the system.

\begin{thm} \label{thm_10}  (\textbf{Storage Liveness}): Every read/write operation executed by a correct client eventually completes.
 \end{thm}
 \begin{proof}
Consider a read/write operation $o$ executed by a correct client $c$. 
For each phase \textit{ph} from $o$, $c$ sends a message to all servers in its current view $cv_c$ and waits for a 
quorum of responses (lines 5,14,22 and 32 -- Algorithm~\ref{alg-client}; lines 1, 5 an 6 -- Algorithm~\ref{alg-server}). In this scenario, we have three cases.

\begin{enumerate}

\item Case each server $i$ of this quorum replies with a view $v_i = cv_c$ (lines 7,15,24 and 33 -- Algorithm~\ref{alg-client}), $c$ ends \textit{ph} in $cv_c$ 
since $cv_c$ has at least a quorum of correct servers (Assumption~\ref{f_lim}). Consequently, $c$ starts the execution of the next phase for $o$, and so on, 
until it finishes the last phase for $o$ in $cv_c$, finally concluding~$o$. 

\item Case some server $i$ of this quorum replies with a view $v_i \neq cv_c$, then $c$ updates $cv_c$ if it 
discovers a most up-to-date view (lines 8,16,25 and 34 -- Algorithm~\ref{alg-client}) and reinitiates \textit{ph}. 
From Assumption~\ref{finite_rec} and Lemma~\ref{lemma_rec_terminacao}, there is a time $t$ 
such that the system stops executing reconfigurations. Consequently, the last reconfiguration executed in the system lead to the installation of some up-to-date 
view $w$ ($V(t) = w$) (Lemma \ref{lemma_rec_terminacao}). Since servers send $w$ to $c$ that updates its $cv_c$ to $w$, we reach the previous Case $1$. 

\item Case $o$ is concurrent with a reconfiguration and servers blocked its execution at line 11 of Algorithm~\ref{alg_recon}, 
then some up-to-date view is going to be installed (Lemma~\ref{lemma_rec_terminacao}) and 
the servers unlock its execution at line 25 of Algorithm~\ref{alg_recon}. Since $cv_c$ is not up-to-date we reach the previous Case $2$.
\end{enumerate}
Therefore, in the three possible cases $o$ eventually completes and the theorem follows.
% $\strut\hfill \Box$
\end{proof}

\begin{lem} \label{lemma_liv_2} 
Consider that a correct server $i$ executes a \textit{join}/\textit{leave} operation at time $t$, by sending an \textit{update} $u$ 
to the servers in $V(t)$. Then, there will be a view $V(t')$, such that $V(t') \neq V(t)$ and $t'>t$, where $u \in V(t')$.

\end{lem}
\begin{proof}
Considering that a correct server $i$ send $u$ to the servers in $V(t)$ at $t$, two cases are possible. 

\begin{enumerate}
\item Case all servers in $V(t)$ start their view generators associated with $V(t)$ proposing a sequence with a view that contains $u$ (line 7 -- Algorithm~\ref{alg_recon}). 
In this case, the fact that all proposals contains $u$ ensures that there will be a view $V(t')$, such that $t' > t$ and $V(t') \neq V(t)$,  where $u \in V(t')$.

\item Case some servers in $V(t)$ start their view generators associated with $V(t)$, before the receipt of $u$, proposing a sequence with a view that does not contain $u$ (line 7 -- Algorithm~\ref{alg_recon}). 
In this case, considering a time $t' > t$ when a reconfiguration has installed a new view $V(t')$, two cases are possible: (a) $u \in V(t')$ since the proposals from the servers that 
received and proposed $u$ was computed in $V(t')$; or (b) $u \not\in V(t')$ and $u$ was sent to all servers in $V(t')$ (lines 12 and 17 -- Algorithm~\ref{alg_recon}) that in the next reconfiguration will start its view generators with a proposal that contains $u$, reducing this  situation to Case $1$.
\end{enumerate}

Therefore, in any case there will be a view $V(t')$, such that $V(t') \neq V(t)$ and $t'>t$, where $u \in V(t')$. Consequently, the lemma follows.
%  $\strut\hfill \Box$
 \end{proof}

\begin{thm} \label{thm_11}  (\textbf{Reconfiguration -- Join Liveness}): 
Eventually, the \emph{enable operations} event occurs at eve\-ry correct server that invoked a $\mathit{join}$ operation.
 \end{thm}
 \begin{proof}
 Consider that a correct server $i$ executes a \textit{join} operation at time $t$, by sending an \textit{update} $u =  \langle +,i\rangle$ to 
the servers in $V(t)$. From Lemma~\ref{lemma_liv_2}, there will be a view  $V(t')$,  such that $V(t') \neq V(t)$ and $t'>t$, 
where $u \in V(t')$. Consequently, when  $V(t')$ is installed at server $i$, the \emph{enable operations} event occurs at 
this server, since the required condition ($i \in V(t')$) is satisfied for servers that are joining the system (lines 14 and 19 -- Algorithm~\ref{alg_recon}).
% $\strut\hfill \Box$
\end{proof}

\begin{thm} \label{thm_12} (\textbf{Reconfiguration -- Leave Liveness}):
Eventually, the \emph{disable operations} event occurs at eve\-ry correct server that invoked a $\mathit{leave}$ operation.
% If a server $j$ installs a view $v$ such that $i \not\in v \wedge (\exists v' : i \in v' \wedge v' \subset v)$, then server $i$ has invoked the $\mathit{leave}$ operation.
 \end{thm}
 \begin{proof}
 Consider that a correct server $i$ executes a \textit{leave} operation at time $t$, by sending an \textit{update}  $u =  \langle -,i\rangle$ 
to the servers in $V(t)$. From Lemma~\ref{lemma_liv_2}, there will be a view $V(t')$, such that $V(t') \neq V(t)$ and $t'>t$, 
where $u \in V(t')$. Consequently, when  $V(t')$ is installed at server $i$, the \emph{disable operations} event occurs at this server, 
since the required condition ($i \not\in V(t')$) is satisfied for servers that are leaving the system (lines 14 and 27 -- Algorithm~\ref{alg_recon}). 
% $\strut\hfill \Box$
\end{proof}

\subsubsection{Safety}

Let us now show a sketch of proof that \textsc{FreeStore} satisfies the safety properties. Firstly, we show that \textsc{FreeStore} implements 
the safety properties of an atomic register~\cite{Lam86}, according to the following definition:

\begin{definition}[Atomic register] Consider two read ope\-rations $r_1$ and $r_2$ executed by correct clients. 
Consider that $r_1$ terminates before $r_2$ initia\-tes. If $r_1$ reads a value $\alpha$ from register ${\cal R}$, then either $r_2$ reads $\alpha$ or $r_2$ reads 
a more up-to-date value than $\alpha$.
\label{def_register}

\end{definition}

Let ${\cal R}$ be the register, and $\omega (\beta)$ be the operation to write the value $\beta$ in ${\cal R}$ with an associated timestamp $ts(\beta)$. Consider that $V(t)$ is the most up-to-date view installed in the system.

\begin{lem} \label{lemma_saf_view}
Each phase of a read/write operation always finishes in the most updated installed view.
\end{lem}

\begin{proof}
For a client $c$ to execute a phase $ph$ of a r/w operation, it is necessary that a quorum of correct servers be in the same view $v$ that $c$ holds 
(lines 7, 15, 24 and 33 -- Algorithm~\ref{alg-client}; lines 1, 5 an 6 -- Algorithm~\ref{alg-server}). 
Thus, we have to show that for any time $t$, $v = V(t)$. 
For the sake of contradiction, consider that at time $t$, $c$ finishes $ph$ by receiving a quorum of replies from servers 
in $v = V(t') : V(t') \neq V(t) \wedge t' < t$. Without loss of generality, consider 
that $V(t')$ is the view immediately before $V(t)$ in the sequence of installed views (Lemma \ref{lemma_rec_sequencia}). 
Consequently, at least $V(t').q$ servers from $V(t')$ do not executed line 11 of Algorithm~\ref{alg_recon} for instalation of $V(t)$ since $V(t') \subset V(t)$ 
and they will stop to execute r/w operation until 
the update of their views to $V(t)$ (lines 11, 18 and 25 -- Algorithm~\ref{alg_recon}). 
However, to install $V(t)$ it is necesssary that at least $V(t').q$ servers from $V(t')$ execute line 12 of Algorithm~\ref{alg_recon} in order to update the state of servers in $V(t)$ 
(line 15 -- Algorithm~\ref{alg_recon}). Cleary, by quorum intersection properties, we reach a contradiction. 
Consequently, we have that $v = V(t)$ and the lemma follows.
% $\strut\hfill \Box$
\end{proof}

\begin{lem} \label{lemma_saf_view_1}
Read/write operations always finish in the most updated installed view.
\end{lem}

\begin{proof}
By Lemma~\ref{lemma_saf_view}, the last phase of a read/write operation $o$ finishes in the most updated installed view. Consequently, $o$ also finishes in the most updated installed view. 
% $\strut\hfill \Box$
\end{proof}

\begin{lem} \label{lemma_saf_1} 
Consider that $\alpha$ is the value from the last write  $\omega(\alpha)$ completed in $V(t)$. Then, in the next reconfigu\-ration $F$ from $V(t)$ 
to $V(t')$, $t' > t$, one of two cases may happen: 

$(1)$ If there is no concurrent write operation with $F$, then $\alpha$ is propagated to $V(t')$;
 
$(2)$ If $F$ is concurrent with a write operation  $\omega(\beta): ts(\beta) > ts(\alpha)$, then either $\alpha$ or $\beta$ is propagated to $V(t')$.
\end{lem}

\begin{proof}
Consider that the last write $\omega(\alpha)$ completed in $V(t)$ has stored the value $\alpha$ in ${\cal R}$. 

\begin{enumerate}
 \item In case there is no concurrent write operation with $F$: at least $V(t).q$ servers store $\alpha$ and $ts(\alpha)$ is the greatest timestamp 
 in the system (otherwise $\omega(\alpha)$ would not be the last). 
From the quorums intersection properties and since $\mbox{STATE-UPDATE}$ messages are sent (lines 12 and 15 -- Algorithm \ref{alg_recon}) and views installed in the system 
form an unique sequence (Lemma \ref{lemma_rec_sequencia}), the servers in $V(t')$ are going to receive $\langle\alpha$, $ts(\alpha)\rangle$ 
coming from $V(t)$ updating their register values (line 16 -- Algorithm \ref{alg_recon}).

\item In case $F$ is concurrent with a write operation  $\omega(\beta): ts(\beta) > ts(\alpha)$: Three situations need to be considered:
 \begin{enumerate}
 \item No servers in $V(t)$ updated their local state with $\langle\beta,ts(\beta)\rangle$ before the instalation of $V(t')$ 
(lines 2-4 -- Algorithm~\ref{alg-server}) and we reduce this case to the Case~$1$.
\item Some servers in $V(t)$ already store $\langle\beta,ts(\beta)\rangle$ sending this value to servers in $V(t')$, consequently it is possible that  some servers from $V(t')$ receive only $\alpha$ while some others receive $\beta$ (line 16 -- Algorithm \ref{alg_recon}). 
Anyway, $\omega(\beta)$ is going to be reinitialized in $V(t')$ (Lemmata~\ref{lemma_saf_view} and \ref{lemma_saf_view_1}) and all its servers will eventually store  $\langle\beta$, $ts(\beta)\rangle$ in ${\cal R}$.

\item All servers in $V(t)$ already store $\langle\beta,ts(\beta)\rangle$ sending this value to servers in $V(t')$, 
consequently all servers in $V(t')$ will store  $\langle\beta$, $ts(\beta)\rangle$ in ${\cal R}$ (line 16 -- Algorithm \ref{alg_recon}). 
\end{enumerate}

\end{enumerate}

Since there is always a sequence of installed views (Lemma~\ref{lemma_rec_sequencia}), 
the above cases are valid on the update of the next view $V(t')$ whenever a subsequent view $V(t'')$,  $t'' > t'$, is installed and so on. Consequently, the lemma follows.  
% $\strut\hfill \Box$
\end{proof}

\begin{lem} \label{lemma_saf_new3}
Consider that a read operation $r_1$ returns a value $\alpha_1$ at time $t_1^{e}$, which has an associated timestamp $ts_1$. 
A read operation $r_2$ started at time $t_2^{s} > t_1^{e}$ returns a value  $\alpha_2$ associated with a 
timestamp $ts_2$ such that either (1) $ts_2 = ts_1$ and $\alpha_2 = \alpha_1$ or (2) $ts_2 \geq ts_1$ and $\alpha_2$ 
was written after $\alpha_1$.

\end{lem}

\begin{proof}
Consider that $r_1$ starts at time $t_1^{s}$ and that $r_2$ ends at time $t_2^{e}$, then $t_2^{e} > t_2^{s} > t_1^{e} > t_1^{s}$. 
% Since $r_1$ finishes at time $t_1^{f}$, it reads $\alpha_1$ from view $V(t_1^{f})$ (Lemma~\ref{lemma_saf_view}).
 Depending on the occurrence of a reconfiguration $F$ that installs a following view $V(t')$ in the unique sequence of 
 intalled views (Lemma \ref{lemma_rec_sequencia}), we have the following possible, mutually exclusive, scenarios:
\begin{enumerate}

\item\label{case1} No concurrent reconfiguration. In this scenario we have no reconfiguration and thus $V(t_1^{s}) = V(t_1^e) = V(t_2^{s}) = V(t_2^{e})$.
Therefore, the behavior of the algorithm is the same of the basis protocol.
Consequently, by the quorum intersection properties, the \textit{write-back} phase from $r_1$ (lines 27-35 -- Algorithm \ref{alg-client}) ensures that $r_2$ reads $\langle \alpha_2, ts_2\rangle$ from servers in $V(t_2^e)$ such that either (1) $ts_2 = ts_1$ and $\alpha_2 = \alpha_1$ or (2) $ts_2 \geq ts_1$ and $\alpha_2$ was written after $\alpha_1$.

% \item\label{case2}$t' < t_1^s$. In this case $V(t')$ is installed before $r_1$ had started. 
% By Lemma~\ref{lemma_saf_view}, both $r_1$ and $r_2$ finish by reading from servers in 
% $V(t') = V(t_1^{s}) = V(t_1^f) = V(t_2^{s}) = V(t_2^{f})$ and we reach the previous Case~\ref{case1}. 

\item\label{case2} Reconfiguration concurrent with $r_1$ ($t_1^s \leq t' < t_1^e$)\footnote{Notice that it is impossible to have $t_1^e = t'$ or $t_2^e = t'$ since r/w operations are stopped while the reconfiguration installs $V(t')$ (line 11 -- Algorithm~\ref{alg_recon}) and, 
 consequently, we can not have a view installation and a r/w operation finishing at the same time.}.
  Since $t_2^s > t_1^e > t'$, the reconfiguration installs $V(t')$ before both $r_2$ starts and $r_1$ ends. Consequently, 
  from Lemmata~\ref{lemma_saf_view} and \ref{lemma_saf_view_1}, both $r_1$ and $r_2$ finish by reading from servers in 
  $V(t') = V(t_1^e) = V(t_2^s) = V(t_2^e)$. For the same reasons of Case~\ref{case1},
  $r_2$ reads $\langle \alpha_2, ts_2\rangle$ from servers in $V(t_2^e)$
such that either (1) $ts_2 = ts_1$ and $\alpha_2 = \alpha_1$ or (2) $ts_2 \geq ts_1$ and $\alpha_2$ 
was written after $\alpha_1$.
  
  \item\label{case3} Reconfiguration between $r_1$ and $r_2$ ($t_1^e < t' < t_2^s$). In this case, the reconfiguration installs $V(t')$ after $r_1$ 
  had finished and before $r_2$ starts. By Lemma~\ref{lemma_saf_1} and 
  the \textit{write-back} phase from $r_1$ (lines 27-35 -- Algorithm \ref{alg-client}), $\langle \alpha_1, ts_1\rangle$ or a more up-to-date pair 
  $\langle \alpha_x, ts_x\rangle$, with $ts_x > ts_1$, is propagated to $V(t')$. By Lemmata~\ref{lemma_saf_view} and \ref{lemma_saf_view_1}, $r_2$ ends by reading from servers in 
$V(t') = V(t_2^{s}) = V(t_2^{e})$. By the quorum intersection properties, $r_2$ reads $\langle \alpha_2, ts_2\rangle$ from servers in $V(t_2^e)$
such that either (1) $ts_2 = ts_1$ and $\alpha_2 = \alpha_1$ or (2) $ts_2 \geq ts_1$ and $\alpha_2$ 
was written after $\alpha_1$.

\item\label{case4} Reconfiguration concurrent with $r_2$ ($t_2^s \leq t' < t_2^e$). %In this case, the \textit{write-back wb} phase of $r1$ ensures that at least a quorum of $V(t).q$ servers from $V(t)$ are going to store $\alpha$ and its timestamp $ts$ (Lemma~\ref{lemma_saf_2}, Case $1$). Any  other concurrent write operation (Lemma~\ref{lemma_saf_2}, Case $2$)  or subsequent to \textit{wb} will only write an associated value with a timestamp which is greater than $ts$. Then, from the intersection property of quorums,  from Lemma~\ref{lemma_saf_3} (Cases $3$ and $4$) and, knowing that $r2$ starts after \textit{wb}, it is ensured that $r2$ will read $\alpha$ or a more up-to-date value written in $V(t)$ or $V(t')$.
  In this case, the reconfiguration installs $V(t')$ after $r_1$ had finished, after $r_2$ has been started and before $r_2$ has been terminated. 
  By Lemma~\ref{lemma_saf_1} and the \textit{write-back} phase from $r_1$ (lines 27-35 -- Algorithm \ref{alg-client}), $\langle \alpha_1, ts_1\rangle$ or a more up-to-date pair 
  $\langle \alpha_x, ts_x\rangle$, with $ts_x > ts_1$, is propagated to $V(t')$.
By Lemmata~\ref{lemma_saf_view} and \ref{lemma_saf_view_1}, $r_2$ ends by reading from servers in 
$V(t') = V(t_2^{e})$. By the quorum intersection properties, $r_2$ reads $\langle \alpha_2, ts_2\rangle$ from servers in $V(t_2^e)$
such that either (1) $ts_2 = ts_1$ and $\alpha_2 = \alpha_1$ or (2) $ts_2 \geq ts_1$ and $\alpha_2$ 
was written after $\alpha_1$.
  
%  \item\label{case5} $t_2^f < t'$. In this case, $F$ installs $V(t')$ after both $r_1$ and $r_2$ has been completed and we fall in the previous Case~\ref{case1}.

\end{enumerate}

Since there is always an unique sequence of installed views (Lemma~\ref{lemma_rec_sequencia}), the above cases are valid whenever a subsequent view $V(t'')$,  $t'' > t'$, is installed and so on. This suffices to prove the lemma.
% $\strut\hfill \Box$
\end{proof}

\begin{thm} \label{thm_1}  (\textbf{Storage Safety}): The read/write protocols satisfy the safety properties of an atomic read/write register.
\end{thm}

\begin{proof} 
This proof follows directly from Lemma~\ref{lemma_saf_new3}.
\end{proof}

 \begin{thm} \label{thm_2} \textbf{Reconfiguration -- Join Safety}: 
% If a server $j$ installs a view $v$ such that $i \in v$, then server $i$ invoked the $\mathit{join}$ operation.
If a server $j$ installs a view $v$ such that $i \in v$, then server $i$ has invoked the $\mathit{join}$ operation or $i$ is member of the initial view.
 \end{thm}
 
 \begin{proof}
Consider that a server $j$ installs a view $v$ such that $i \in v$, then some server $z$ started its view generator associated with a previous view $v' \subset v$ 
proposing a view sequence containing $v$ (Lemma \ref{lemma_rec_vivo_vAuxiliar}) such that the update $\langle +, i \rangle \in v$. In this case, 
$\langle +, i \rangle \in \mathtt{RECV}$ at $z$ (line 7 -- Algorithm \ref{alg_recon}). Consequently, $i$ sent a message to $z$ with the join request by executing the $join$ operation (lines 1-2 -- Algorithm \ref{alg_recon}) or $i$ is member of the initial view.
% $\strut\hfill \Box$
\end{proof}

\begin{thm} \label{thm_3}  \textbf{Reconfiguration -- Leave Safety}: 
% If a server $j$ installs a view $v$ such that $i \not\in v$, then server $i$ invoked the $\mathit{leave}$ operation.
If a server $j$ installs a view $v$ such that $i \not\in v \wedge (\exists v' : i \in v' \wedge v' \subset v)$, then server $i$ has invoked the $\mathit{leave}$ operation.
\end{thm} 

 \begin{proof}
Consider that a server $j$ installs a view $v$ such that $i \not\in v$  and that $\exists v' : i \in v' \wedge v' \subset v$, then some server $z$ started its view generator 
associated with a previous view proposing a view sequence containing a view 
$v$ (Lemma \ref{lemma_rec_vivo_vAuxiliar}) such that the update $\langle -, i \rangle \in v$. 
In this case, $\langle -, i \rangle \in \mathtt{RECV}$ at $z$ (line 7 -- Algorithm \ref{alg_recon}). Consequently, $i$ sent a 
message to $z$ with the leave request by executing the $leave$ operation (lines 3-4 -- Algorithm \ref{alg_recon}).
% $\strut\hfill \Box$
\end{proof}

\section{Discussion}
\label{discussao}

% \vspace{-1mm}
\subsection{DynaStore vs. \textsc{FreeStore}}
\label{comparacao}

% (\textbf{Alysson: VERIFICAR})
This section discusses some differences between DynaStore~\cite{Agu11} and \textsc{FreeStore}.
% Although our focus is on comparing our approach with DynaStore, when adequate, we comment the relationship between our work and the recent proposed SpSn~\cite{Gaf15} and SmartMerge~\cite{Jen15}.
Although our focus is on comparing our approach with DynaStore, we comment the relationship between these works and the recent proposed SpSn~\cite{Gaf15} and SmartMerge~\cite{Jen15}. 

\subsubsection{Convergence Strategy} 
% \paragraph{Convergence Strategy} 

In DynaStore, the reconfiguration process generate a graph of views through which it is possible to identify a sequence of \emph{established} views (see Figure~\ref{fig:comparison}, left). 
A view that is not established works as an \emph{auxiliary view}, that must be accessed during r/w operations. 
For any established view $v$, the maximum number of views that can succeed it is $|v|$, i.e, each member of $v$ can generate a new view, moreover, new views representing 
the combinations of these views could also be generated. 
In contrast, reconfigurations in \textsc{FreeStore} install only a single sequence of views (see Figure~\ref{fig:comparison}, right). 
In this case, different generated sequences of views are organized in an unique sequence of installed views.
For each installed view $v$, our implementation of $\mathcal{L}$ bounds the number of generated view sequences to $|v| - v.q+ 1$. 
SpSn and SmartMerge also install a sequence of ordered configurations (views).
% 
% 
% %%%%%%%%%%%%%%%%%%%%%%%%%%%%%%%%%%%%%%%
% \begin{figure}[!h]
% %\vspace{-2mm}
% \begin{center}
% \subfigure[\footnotesize{Dynastore.}]{\label{fig:dyn}{\includegraphics[width=0.29\columnwidth]{figs/dynastore}}}
% % \vspace{-0.2cm}
% \subfigure[\footnotesize{\textsc{FreeStore}.}]{\label{fig:free}{\includegraphics[width=0.3\columnwidth]{figs/freestore}}}
% \end{center}
% % \vspace{-0.3cm}
% \caption{DynaStore vs. \textsc{FreeStore} convergence approaches.}
% \label{fig:dynvsfree}
% %  \vspace{-4mm}
% \end{figure}
% %%%%%%%%%%%%%%%%%%%%%%%%%%%%%%%%%%%%%%%

\subsubsection{Liveness}

% \paragraph{Liveness}
% % \label{live_diss}
% 
In DynaStore, a process executes a leave and halts the system.
 However, for any time $t$, there is a bound on the number of processes that can leave the system without compromising liveness.
 Let $F(t)$ be the set of processes that crashed until time $t$ and $J(t)$/$L(t)$ the set of pending joins/leaves at $t$, respectively.
 The liveness condition of DynaStore states that, fewer than $|V(t)|/2$ processes out of $V(t) \cup J(t)$ should be in $F(t) \cup L(t)$~\cite{Agu11}. 
 In contrast, in \textsc{FreeStore} a process that executes a leave should wait for the installation of the updated view (without itself).
 If this restriction is not respected, the leaving server is considered faulty.
 This approach specifies a bound on the number of leaves for \textsc{FreeStore}: fewer than $|V(t)|/2$ processes out of $V(t)$ should be in $F(t) \cup L(t)$. % (see Section~\ref{dyn_ass}). 
% % Finally, it is worth to notice that the DynaStore protocol does not specify when it is safe for the processes that implement the weak snapshot objects to leave the system, since these objects may be accessed by late processes to reach the last established view.
SpSn does not specify a liveness condition and SmartMerge uses policies to prevent the generation of unsafe views, which in practice restricts the number of allowed leaves.

% \vspace{-2mm}
% \subsubsection{Normal Case Execution}
\subsubsection{Normal Case Execution}
\label{df_nce}

In DynaStore, the reconfiguration process generate a graph of views through which it is possible to identify a sequence of \emph{established} views. 
 A view that is not established works as an \emph{auxiliary view}, but must be accessed during r/w operations. 
For each view (established or not), DynaStore associates a \textit{weak snapshot object} that is used to store updates.
For any view $v$, its weak snapshot object $wso_v$ is supported by a set $S_v$ of $|v|$ static SWMR registers, i.e., one register for each member of $v$. 
During a r/w operation on $v$, $wso_v$ must be accessed twice to verify if some update was executed on $v$.
Each access to $wso_v$ comprises two reads in each register of $S_v$.
Thus, to execute a r/w on $v$, it is necessary a total of $4|v|$ ($4$ sequential, $|v|$ parallel) quorum accesses, with two communication steps each. 
SpSn and SmartMerge use a similar approach: before execute each r/w a set with $|v|$ SWMR registers must be accessed to check for updates. 
In contrast, the overhead introduced by \textsc{FreeStore} on a r/w is the local verification if the client' view is equal to the current view of the servers.

% \vspace{-2mm}
% \subsubsection{R/W and Reconfiguration Concurrency}
\subsubsection{R/W and Reconfiguration Concurrency}
A r/w operation needs to traverse the graph of views generated by DynaStore to find a view where it is safe to be executed (the most up-to-date established view).
During this processing, each edge of the graph must be accessed in order to verify if there is a more up-to-date view.
For each view $v$, it is necessary to access its weak snapshot object, which requires $2|v|$ ($2$ sequential, $|v|$ parallel) quorum accesses. 
While the system is converging to some established view, the r/w operation does not terminate.
Similarly, in SpSn and SmartMerge, a r/w concurrent with a reconfiguration also needs to access intermediary views to found the most up-to-date installed view. 
On the other hand, in \textsc{FreeStore}, r/w operations are directed to the most up-to-date installed view $v$, without accessing any auxiliary view.
However, after some sequence of views for updating $v$ is obtained, the r/w operations are stopped and can only terminate after the installation of the most up-to-date view of this sequence. 

% (\textbf{Alysson: ESSE COMENTÁRIO DO ULTIMO PARÁGRAFO É BEM IMPORTANTE, SERÁ QUE NÃO DEVIA SER MUITO MAIS COMENTADO?})

% \vspace{-2mm}
% \subsubsection{Performance of Read/Write Operations}
\subsection{Performance of Read/Write Operations}

Table~\ref{table:comparacao} shows the number of communication steps demanded to execute a r/w operation with DynaStore (DS)~\cite{Agu11}, SmartMerge (SM)~\cite{Jen15}, SpSn~\cite{Gaf15} and \textsc{FreeStore} (FS), for a process that handles an updated or outdated view.
% When a process $p$ that handles an up-to-date view $v$ wants to execute a r/w operation in DynaStore it needs 2 accesses to $wso_v$, which requires 8 communication steps to execute 4 reads in each register in $S_v$; and 4 communication steps to read/write the value from/to $v$.
% In \textsc{FreeStore}, this operation would require the same number of communication steps of a static r/w protocol (Section \ref{rw_prt}).
A r/w operation in \textsc{FreeStore} requires at most a third of the number of communication steps required to execute it in DynaStore or SpSn and at most a half of the steps required in SmartMerge.
More important, it matches the performance of static protocols in the absence of reconfigurations. Moreover, some steps of DynaStore, SmartMerge and SpSn are to access a set of static SWMR registers, while 
in \textsc{FreeStore} each step acesses only a single register, demanding few messages.

% \begin{table}[!ht]
% \begin{center}
% \footnotesize{
% \begin{tabular}{|c|cc|cc|}
%   \hline
%     & \multicolumn{2}{c}{\emph{Updated View}} & \multicolumn{2}{|c|}{\emph{Outdated View}}\\
% %  \hline
%       &  DynaStore & \textsc{FreeStore}  & DynaStore & \textsc{FreeStore}\\
%   \hline \hline
%     \emph{Read} & 12 & 2/4 & 19 & 4/6 \\
%        
%   \hline
%     \emph{Write}  & 12 & 4 & 19 & 6 \\
%       
%   \hline
% 
% \end{tabular}
% %  \end{spacing}
%  }
% \end{center}
% % \vspace{-8mm}
% \caption{Number of communication steps of r/w operations.}
% \label{table:comparacao}
% % \vspace{-5mm}
% \end{table}

% Table~\ref{table:comparacao} also shows numberss for a process $p$ that handles an outdated view $v$.
% DynaStore needs 6 communication steps to access each edge of the graph of views.
% To simplify our analysis, we consider that only one edge need to be accessed to find the up-to-date view $w$ (the best case). 
% In this case, $p$ needs 6 communication steps to find $w$, 12 communication steps to execute a r/w operation in $w$ and 1 communication step to notify other processes about the up-to-date view.
% In \textsc{FreeStore}, 2 additional steps are required to get $w$ from the servers in $v$ since one phase of the protocol is restarted.

\vspace{-1mm}
\begin{table}[!ht]
 \centering
 \caption{Communication steps of r/w operations.}
 \vspace{-1mm}
 \label{table:comparacao}

\footnotesize{
\begin{tabular}{|c|cccc|cccc|}
  \hline
    & \multicolumn{4}{c}{\emph{Updated View}} & \multicolumn{4}{|c|}{\emph{Outdated View}}\\
%  \hline
      &  DS & SM & SpSn & FS  & DS & SM & SpSn & FS\\
  \hline \hline
    \emph{Read} & $12$ & 8 & $14$ & $2/4$ & $19$ & 16 & $22$  &  $4/6$ \\
       
  \hline
    \emph{Write}  & $12$ & 8 & $18$  & $4$ & $19$ & 16 & $26$  & $6$ \\
      
  \hline

\end{tabular}
%  \end{spacing}
 }

\end{table}
\vspace{-5mm}

\subsection{Consensus-free vs. Consensus-based Reconfiguration}

% In this section we study the performance of consensus-free and consensus-based reconfiguration protocols for atomic fault-tolerant storage.
Table \ref{table:approaches} shows the number of communication steps required for processing a reconfiguration using consensus-free (DynaStore, SmartMerge, SpSn and \textsc{FreeStore} with $\mathcal{L}$) and consensus-based (RAMBO and \textsc{FreeStore} with $\mathcal{P}$) algorithms. 
We present the number of communication steps for the best case scenario -- when all processes propose the same updates in a reconfiguration -- and for the worst case -- 
when each process proposes different updates.
In order to simplify our analysis, we consider that no reconfiguration is started concurrently with the one we are analyzing.
Similarly, we assume a synchronous execution of the Paxos protocol~\cite{Lam98}, which requires only three communication steps, for consensus-based algorithms.

\vspace{-1mm}
\begin{table}[!ht]
 \centering
 \caption{Communication steps for reconfiguration.}
 \vspace{-1mm}
 \label{table:approaches}
 \footnotesize{
\begin{tabular}{|c|c|c|}
  \hline
    \emph{Consensus-free approaches}  &  \emph{Best Case} & \emph{Worst Case} \\ \hline
    \hline 
    DynaStore~\cite{Agu11} & $23$  & $18|v| + 5$  \\ \hline
    SmartMerge~\cite{Jen15}  & $11$ &  $6|v|+5$ \\ \hline
    SpSn~\cite{Gaf15}  & $14$ &  $8|v|+6$ \\ \hline
    \textsc{FreeStore} with $\mathcal{L}$  & $4$ &  $7|v| - 2v.q - 1$ \\ \hline
%    \textsc{FreeStore} with $\mathcal{L}$ (optimized)  & 3 & $5|v| - v.q - 1$ \\ \hline
    \hline \hline
    \emph{Consensus-based approaches}  &  \emph{Best Case} & \emph{Worst Case} \\ \hline
    \hline
    RAMBO~\cite{Gil10} & $7$ & $7$  \\ \hline
    \textsc{FreeStore} with $\mathcal{P}$  & $5$ & $5$ \\ \hline
\end{tabular}
}
     \vspace{-2mm}
\end{table}

\textsc{FreeStore} reconfiguration is significantly more efficient than previous consensus-based and consensus-free protocols.
\textsc{FreeStore} with $\mathcal{P}$ requires two less communication steps than RAMBO and, more importantly, \textsc{FreeStore} with $\mathcal{L}$ outperforms other consensus-free approaches by almost an order of magnitude.
In particular, \textsc{FreeStore} with $\mathcal{L}$ presents the best performance among all considered reconfiguration protocols in the best case, which is expected to be the norm in practice.
An open question is if this number constitutes a lower bound.

\section{Conclusions}

This paper presented a new approach to reconfigure fault-tolerant storage systems, which clarify the differences between relying or not to consensus for agreement in the next view to be installed. 
The main result is a protocol that is simpler and cheaper (in terms of communication steps for either r/w operations or reconfigurations) than already proposed solutions 
and that fully decouples the execution of r/w and reconfigurations. %, imposing a negligible overhead to the static ABD protocol.
Other interesting result is that in the best case, consensus-free reconfigurations are cheaper than the consensus-based ones.

A final contribution of this work is the introduction of a new abstraction called view generator.
We believe that exploring different instantiations of this abstraction and their properties is an important avenue for future work.

% % \begin{spacing}{0.85}
% \bibliographystyle{IEEEtran}
% \bibliography{referencias}
% % \end{spacing}

\end{document}